\documentclass[ reqno, 11pt]{amsart}

\usepackage{fullpage}
\usepackage[pass]{geometry}

\usepackage{lscape}
\usepackage{cases}
\usepackage{latexsym}
\usepackage{amsmath}
\usepackage{csquotes}
\usepackage[arrow,matrix]{xy}
\usepackage{stmaryrd}
\usepackage{amsfonts}
\usepackage{amsmath,amssymb,amscd,bbm,amsthm,mathrsfs,dsfont}
\usepackage[
colorlinks=true,         % Enable colored links
linkcolor=cyan,          % Color for internal links
urlcolor=cyan,           % Color for external links (URLs)
citecolor=cyan,          % Color for citation links
filecolor=cyan,          % Color for file links
menucolor=cyan           % Color for menu links (optional)
]{hyperref}
\usepackage{mathtools,amsthm}
\newtheoremstyle{case}{}{}{}{}{}{:}{ }{}
\theoremstyle{case}

\usepackage{tikz}

\usepackage{fancyhdr}
\usepackage{amsxtra,ifthen}
\usepackage{verbatim}
\usepackage{calrsfs}
\DeclareMathAlphabet{\pazocal}{OMS}{zplm}{m}{n}

\numberwithin{equation}{section}

\theoremstyle{plain}
\newtheorem{theorem}{Theorem}[section]
\newtheorem{lemma}[theorem]{Lemma}
\newtheorem{proposition}[theorem]{Proposition}

\newtheorem{corollary}[theorem]{Corollary}

\theoremstyle{definition}
\newtheorem{definition}[theorem]{Definition}
\newtheorem{example}[theorem]{Example}

\begin{document}
	\newtheorem{thm}{Theorem}[section]
	\newtheorem{cor}[thm]{Corollary}
	\newtheorem{lem}[thm]{Lemma}
	\newtheorem{defn}[thm]{Definition}
	\newtheorem{prop}[thm]{Proposition}
	\newtheorem{exm}[thm]{Example}
	\newtheorem{fact}[thm]{Fact}

	\title{Reversible double cyclic codes over a chain ring\\     %title first line
		%OF PAPERS           %title second line
	}
	
	%\medskip
	
	\thanks{}
	
	\author[M. Anwar, M. R. Mozumder, M. Rashid, M. A. Raza]{Mohd Anwar, Mohd Arif Raza, Mohd Rashid, Muzibur Rahman Mozumder$^{\ast}$}

	\address{Mohd Anwar\newline
		\indent Department of Mathematics\newline
		\indent Aligarh Muslim University \newline 
		\indent Aligarh, India}
	\email{gi3862@myamu.ac.in}

	\address{Mohd Arif Raza\newline
		\indent Department of Mathematics \newline
		\indent College of Science \&\ Arts-Rabigh,\newline
		\indent King Abdulaziz University\newline
		\indent Jeddah, Saudi Arabia}
	\email{mreda@kau.edu.sa}

	\address{Mohd Rashid\newline
		\indent Department of Mathematics\newline
		\indent Aligarh Muslim University \newline 
		\indent Aligarh, India}
	\email{rashidaraz253@gmail.com}

	\address{Muzibur Rahman Mozumder\newline
		\indent Department of Mathematics \newline
		\indent Aligarh Muslim University\newline
		\indent Aligarh, India}
	\email{muzibamu81.maths@amu.ac.in}

	\thanks{$^{\ast}$ Corresponding author}
	
	\maketitle

	%\address{ALI N. A. KOAM and M. A. Ansari, Department of Mathematics,
		%Faculty of Science, Jazan University, Saudi Arabia}
	%\email{\href{mailto:akoum@jazanu.edu.sa}{akoum@jazanu.edu.sa,  maansari@jazanu.edu.sa}}
	
	\date {}

	%\bigskip
	
	%%%%%%%%TEXT %%%%%%%%%%%%%
	
	\begin{abstract}
		In this paper, we study the structure of double cyclic codes of length $(\gamma,\delta)$ over $\mathbb F_q+u\mathbb F_q, u^2=0$. We also study the dual of double cyclic code of length $(\gamma,\delta)$ and give a minimal spanning set of double cyclic codes. Moreover, we study the necessary and sufficient conditions for a double cyclic code to be reversible and reversible-complement double cyclic code and with the help of these codes, we constructed DNA codes over  $\mathbb F_4+u\mathbb F_4, u^2=0$. We also constructed some optimal codes to support our results.
	\end{abstract}

	%\vskip 24pt
	
	
	\vskip 10pt
	
	
	\section{\textbf{Introduction}}
	Algebraic coding theory plays a crucial role in easy and efficient transmission of information and reliable storage of data, linear codes are crucial for error-detection and correction during the information transmission and improving the data storage reliability. Therefore, the structure of linear and cyclic codes over finite fields have been extensively studied \cite{huffman}. In 1994, Hammons et al. \cite{hammons} showed that several good non linear codes over $\mathbb Z_2$ can be identified as the Gray images of linear codes over $\mathbb Z_4$. Since then, due to the easy encoding and decoding linear codes over finite rings are intensively investigated. Since cyclic codes over finite rings enjoy a nice algebraic structure and polynomial representation, they are mostly studied among the class of linear codes, for greater insight one can see \cite{abualrub,liang,prakash,alali,ashraf,dinh}. Further, reversible cyclic codes are vital as they reduce the effort in the encoding and decoding process; also, reversibility is an essential property when studying DNA codes. DNA codes are very crucial in coding theory since many combinatorial problem has been addressed by DNA computing such as Maximal clique problem \cite{ouyang}, the Hamiltonian path problem \cite{adleman}, Mansuripur et
	al. \cite{mansuripur} showed the use of DNA codes for storage media, Adleman and coauthors \cite{adleman des} cracked the Data Encryption Standard (DES) cryptosystem using DNA computing techniques etc. Therefore, good error-correcting codes have been constructed using the DNA structure as a model, and error-correcting codes with properties similar to DNA structure have also been utilized for understanding DNA. The linear construction of DNA codes was studied by Gaborit and King \cite{11}. Abualrub et al. studied the DNA codes over the finite field of four elements \cite{12}. Later, Siap et al. discussed DNA codes over the ring $F_2[u]/\langle u^2-1\rangle$ with four elements \cite{13}. DNA codes over the ring $F_2[u]/\langle u^4-1\rangle$ with 16 elements were studied by Yildiz and Siap \cite{10}.
	Liang and Wang \cite{liang} studied cyclic DNA codes over the ring $F_2+uF_2$, $u^2=0$. Subsequently, Mostafanasab and Darani \cite{16} explored the cyclic DNA codes over the ring $R=F_2+uF_2+u^2F_2$, $u^3=0$.
	\par In \cite{borges} Borges et al., introduced a new class of linear codes called double cyclic codes which are the generalization of cyclic codes. The authors studied the structure of double cyclic codes over $\mathbb Z_2$ and their duals as submodules of $\mathbb Z_2[x]$ module $\frac{\mathbb Z_2[x]}{\langle x^r-1\rangle}\times\frac{\mathbb Z_2[x]}{\langle x^s-1\rangle}$, where $r$ and $s$ are any non negative integers. Besides, they also compare $\mathbb Z_2$-double cyclic codes with other families of cyclic codes. After that, Gao et al. \cite{j Gao} extended the study of double cyclic codes over the ring of integer modulo $4$. Double cyclic codes over a chain ring $\mathbb F_q+u\mathbb F_q+u^2\mathbb F_q,\ u^3=0$ are studied by Yao and Shi \cite{yao}. Later Gao and Hou \cite{hou} showed that $\mathbb Z_4$-double cyclic codes are asymptotically good. In the theory of cyclic codes reversible cyclic codes are crucial as they simplify the encoding and decoding process, also reversibility is a crucial property when studying the DNA codes. Therefore, Patanker \cite{patanker} studied the reversibility of double cyclic codes over $\mathbb Z_2$. In \cite{kanlaya} Kanlaya and Klin-Eam extended this study to double cyclic codes over the ring $\mathbb F_2+u\mathbb F_2$ for the construction of DNA codes.
	\par Motivated by aforementioned works, we further extend the study of double cyclic codes over the ring $\mathcal{R}=\mathbb F_q+u\mathbb F_q,\ u^2=0$, where $q$ is some prime power. we determine the structure of double cyclic codes and the minimal spanning set of double cyclic codes over $\mathcal{R}$. We also discuss the dual of double cyclic codes over $\mathcal{R}$. Moreover, we explore reversible double cyclic codes over $\mathcal{R}$. Further, we study reversible-complement double cyclic codes over $\mathcal{R}$ for the applications to DNA codes. The article is organised as follows, In Section \ref{s2} some basic notions and definitions are presented which are required in further sections. Section \ref{s3} presents the structure of double cyclic code over $\mathcal{R}$. Section \ref{s4} and  \ref{s5} discuss minimal spanning set and dual code of double cyclic code. Reversibility of double cyclic codes over $\mathcal{R}$ is discussed in Section \ref{s6}. Cyclic DNA codes are studied in Section \ref{s7}. Finally, we concluded this article in Section \ref{s9}.

	\section{\textbf{Preliminaries}}\label{s2}
	Let $\mathbb F_q$ be a finite field of $q$ elements, where $q=p^e,\ e\in\mathbb{N}$ and $p$ is any prime. Let $\mathcal{R}=\mathbb F_q+u\mathbb F_q,\ u^2=0$ then, $\mathcal{R}$ is a finite commutative ring with unity of characteristic $p$ and cardinality $q^2$. The ring $\mathcal{R}$ is a principal ideal chain ring with unique maximal ideal $\langle u\rangle$ hence $\mathcal{R}$ is a local ring. The Gray map from $\mathcal{R}$ to $\mathbb F^2_q$ is defined as
	\begin{align*}
		\vartheta:\mathcal{R}\to&\mathbb F^2_q\\
		\vartheta(a+ub)=&(a,a+b),\ a,b\in\mathbb F_q.
	\end{align*}
	This map can naturally be extended to $\Phi:\mathcal{R}^n\to \mathbb F^{2n}_q$. Recall that a linear code $\mathbf{C}$ of length $n$ over any commutative ring $\mathbf{R}$ is an $\mathbf{R}$-submodule of $\mathbf{R}$-module $\mathbf{R}^n$. A linear code $\mathbf{C}$ is called cyclic if $\mathbf{C}$ is closed under cyclic shift i.e., for
	$$(c_0,c_1,\ldots,c_{n-1})\in\mathbf{C}\implies(c_{n-1},c_0,\ldots,c_{n-2})\in\mathbf{C}\ \forall(c_0,c_1,\ldots,c_{n-1})\in\mathbf{C}.$$
	For any $b=(b_0,b_1,\ldots,b_{n-1})\in\mathbf{R}^n$, we can identify $b$ by a polynomial $b(x)=b_0+b_1x+\cdots+b_{n-1}x^{n-1}$ in $\frac{\mathbf{R}[x]}{\langle x^n-1\rangle}$. Then, any linear code $\mathbf{C}$ of length $n$ over $\mathbf{R}$ is cyclic if and only if $\mathbf{C}$ is an ideal in $\frac{\mathbf{R}[x]}{\langle x^n-1\rangle}$.
	\par The following result gives the complete structure of cyclic code over the ring $\mathcal{R}$ of length $n$. For further insight readers may refer to \cite{prakash}.
	\begin{theorem}{\em \cite[Theorem 1]{prakash}}\label{th 2.1}
		Let $\mathcal{C}$ be a cyclic code of length $n$ over $\mathcal{R}$.
		\begin{itemize}
			\item[\rm 1.] If $gcd(n,q)=1$, then $\frac{\mathcal{R}[x]}{\langle x^n-1\rangle}$ is a principal ideal ring and $\mathcal{C}=\langle g(x),ua(x)\rangle=\langle g(x)+ua(x)\rangle$, where $g(x),a(x)\in\mathbb F_q[x]$ and $a(x)|g(x)|x^n-1$ mod $q$.
			\item[\rm 2.] If $gcd(n,q)\neq 1$, then
			\begin{itemize}
				\item[\rm (a)] $\mathcal{C}=\langle g(x)+up(x)\rangle$, where $g(x),p(x)\in\mathbb F_q[x]$ and $g(x)|x^n-1$ mod $q$, $(g(x)+up(x))|(x^n-1)$, $g(x)|p(x)(\frac{x^n-1}{g(x)})$ and $g(x)=a(x)$.
				\item[\rm (b)] $\mathcal{C}=\langle g(x)+up(x),ua(x)\rangle$, where $g(x),a(x),p(x)\in\mathbb F_q[x]$ and $a(x)|g(x)|x^n-1$ mod $q$, $(g(x)+up(x))|(x^n-1)$, $a(x)|p(x)(\frac{x^n-1}{g(x)})$ and $deg(g(x))>deg(a(x))>deg(p(x))$.
			\end{itemize}
		\end{itemize}
	\end{theorem}
	Consider that $\gamma$ and $\delta$ are two non-negative integers. Then, $\mathcal{R}^{\gamma+\delta}$ is an $\mathcal{R}$-submodule of $\mathcal{R}^\gamma\times\mathcal{R}^\delta=\mathcal{R}^{\gamma,\delta}$. For any element $c=(b_0,\ldots,b_{\gamma-1}|d_0,\ldots,d_{\delta-1})\in\mathcal{R}^{\gamma,\delta}$, the double cyclic shift of $c$ is defined as
	$$\sigma(c)=(b_{\gamma-1},b_0,\ldots,b_{\gamma-2}|d_{\delta-1},d_0,\ldots,d_{\delta-2}).$$
	\begin{definition}
		A linear code $\mathcal{C}$ of length $n=\gamma+\delta$ over $\mathcal{R}$ is said to be double cyclic code of length $(\gamma,\delta)$ over $\mathcal{R}$ if $\sigma(c)\in\mathcal{C}$ for all $c\in\mathcal{C}$.
	\end{definition}
	We can identify any element $c=(b|d)=(b_0,b_1,\ldots,b_{\gamma-1}|d_0,d_1,\ldots,d_{\delta-1})\in\mathcal{R^{\gamma,\delta}}$ by an element $c(x)$ in $\frac{\mathcal{R}[x]}{\langle x^\gamma-1\rangle}\times\frac{\mathcal{R}[x]}{\langle x^\delta-1\rangle}=\mathcal{R_{\gamma,\delta}}$ as follows:
	$$c(x)=(b(x)|d(x))=(b_0+b_1x+\cdots+b_{\gamma-1}x^{\gamma-1}|d_0+d_1x+\cdots+d_{\delta-1}x^{\delta-1}).$$ This provides a one-to-one correspondence between $\mathcal{R^{\gamma,\delta}}$ and $\mathcal{R_{\gamma,\delta}}$. The ring $\mathcal{R_{\gamma,\delta}}$ is an $\mathcal{R}[x]$-module with respect to usual addition and multiplication $\star$ defined in (\ref{eq2.1}) for any $p(x)\in\mathcal{R}[x]$ and $c(x)\in\mathcal{R_{\gamma,\delta}}$
	\begin{equation}\label{eq2.1}
		p(x)\star c(x)=p(x)\star (b(x)|d(x))=(p(x)b(x)|p(x)d(x)).
	\end{equation}
	Where multiplication $p(x)b(x)$ and $p(x)d(x)$ is done under$\mod x^\gamma-1$ and$\mod x^\delta-1$ respectively. Moreover for $c(x)\in\mathcal{R_{\gamma,\delta}}$ the multiplication $x\star c(x)$ gives double cyclic shift of $c\in\mathcal{R^{\gamma,\delta}}$.
	
	\begin{theorem}\cite{kanlaya}
		A linear code $\mathcal{C}$ of length $n=\gamma+\delta$ over $\mathcal{R}$ is a double cyclic code of length $(\gamma,\delta)$ over $\mathcal{R}$ if and only if $\mathcal{C}$ is an $\mathcal{R}[x]$-submodule of $\mathcal{R_{\gamma,\delta}}$.
	\end{theorem}
	For any two element $c\text{ and }c'$ in $\mathcal{R}^\gamma\times\mathcal{R}^\delta$, the inner product between $c\text{ and }c'$ is defined by
	\begin{align*}
		c\cdot c'=&(b_0,b_1,\ldots,b_{\gamma-1}|d_0,d_1,\ldots,d_{\delta-1})\cdot(b'_0,b'_1,\ldots,b'_{\gamma-1}|d'_0,d'_1,\ldots,d'_{\delta-1})\\
		=&\sum_{i=0}^{\gamma-1} b_ib'_i+\sum_{j=0}^{\delta-1} d_jd'_j\in\mathcal{R}.
	\end{align*}
	Now, we can define the dual code of a double cyclic code $\mathcal{C}$.
	\begin{definition}
		Let $\mathcal{C}$ be a double cyclic code over $\mathcal{R}$. Then dual code of $\mathcal{C}$ is defined as
		$$\mathcal{C}^\perp=\{v\in\mathcal{R}^\gamma\times \mathcal{R}^\delta|c\cdot v=0\ \forall c\in\mathcal{C}\}.$$
	\end{definition}
	
	\begin{lemma}
		Let $\mathcal{C}$ be double cyclic code over $\mathcal{R}$. Then dual code $\mathcal{C}^\perp$ of $\mathcal{C}$ is also a double cyclic code over $\mathcal{R}$ of same length.
	\end{lemma}
	\begin{table}[ht]
		\caption{Correspondence ($\boldsymbol{\tau}$) of DNA base pair with elements of the ring $\mathcal{R}$}
		\centering
		\begin{tabular}{cc|cc}
			Elements of $\mathcal{R}$ & DNA pairs & 	Elements of $\mathcal{R}$   & DNA pairs \\
			$(a)$  &$(\boldsymbol{\tau}(a))$& $(a)$  &$(\boldsymbol{\tau}(a))$\\
			\hline
			0                &   $AA$      &  $\alpha u$       & $CC$    \\
			1                &    $AT$     &  $1+\alpha u$       &  $CG$    \\
			$\alpha$         &   $AC$      &  $\alpha+\alpha u$       & $CA$     \\
			$\alpha^2$       &   $AG$      &  $\alpha^2+\alpha u$       & $CT$     \\
			$u$              &   $TT$      &  $\alpha^2u$       &   $GG$   \\
			$1+u$            &   $TA$      &  $1+\alpha^2u$       &   $GC$   \\
			$\alpha+u$       &   $TG$      &  $\alpha+\alpha^2u$       &  $GT$    \\
			$\alpha^2+u$     &   $TC$      &  $\alpha^2+\alpha^2u$       &  $GA$ \\
			\hline
		\end{tabular}
		\label{ta1}
	\end{table}
	\section{Double cyclic codes over $\mathcal{R}$}\label{s3}
	In this section we discuss the structure of double cyclic code of length $(\gamma,\delta)$ over $\mathcal{R}$ when both $\gamma\text{ and }\delta$ are relatively prime to $q$.
	\par Let $\mathcal{C_\gamma}$ be the coordinate projection of $\mathcal{C}$ on the first $\gamma$ coordinates and $\mathcal{C_\delta}$ be the coordinate projection of $\mathcal{C}$ on the last $\delta$ coordinates then, $\mathcal{C}$ is said to be separable if $\mathcal{C}=\mathcal{C_\gamma}\times\mathcal{C_\delta}$. Now let us suppose that $\mathcal{C}$ be an $\mathcal{R}[x]$-submodule of $\mathcal{R_{\gamma,\delta}}$ and consider the maps defined by
	\begin{align*}
		\phi_\gamma:\mathcal{C}\to&\frac{\mathcal{R}[x]}{\langle x^\gamma-1\rangle}=\mathcal{R_\gamma}\\
		\phi_\gamma(b(x)&|d(x))=b(x)
	\end{align*}
	for all $(b(x)|d(x))\in\mathcal{C}$ and 
	\begin{align*}
		\phi_\delta:\mathcal{C}\to&\frac{\mathcal{R}[x]}{\langle x^\delta-1\rangle}=\mathcal{R_\delta}\\
		\phi_\delta(b(x)&|d(x))=d(x)
	\end{align*}
	for all $(b(x)|d(x))\in\mathcal{C}$. Then $\phi_\gamma$ and $\phi_\delta$ are $\mathcal{R}[x]$- module homomorphisms.
	\par Now, let us explore the structure of double cyclic code over $\mathcal{R}$ when both $\gamma\text{ and }\delta$ are coprime with $q$.
	\begin{theorem}\label{th3.1}
		Let $\mathcal{C}$ be a double cyclic code of length $(\gamma,\delta)$ over $\mathcal{R}$. If both $\gamma$ and $\delta$ are relatively prime with $q$ then, $\mathcal{C}=\langle(g_1(x)+ua_1(x)|0),(t(x)|g_2(x)+ua_2(x))\rangle$, where $g_i(x),a_i(x)\in\mathbb F_q[x]$ for $i=1,2$, $a_1(x)|g_1(x)|x^\gamma-1$ mod $q$ and $a_2(x)|g_2(x)|x^\delta-1$ mod $q$ and $t(x)\in\mathcal{R}[x]$.
	\end{theorem}
	\begin{proof}
		Consider $\mathcal{C}$ be a double cyclic code over $\mathcal{R}$ of length $(\gamma,\delta)$. Let $\phi_\delta$ be the map as defined above in this section. Then $\phi_\delta$ is an $\mathcal{R}[x]$-module homomorphism and $\phi_\delta(\mathcal{C})$ is an ideal in $\mathcal{R_\delta}$. Then from Theorem \ref{th 2.1}
		$$\phi_\delta(\mathcal{C})=\langle g_2(x)+ua_2(x)\rangle,\text{ where }g_2(x),a_2(x)\in\mathbb F_q[x]\text{ and }a_2(x)|g_2(x)|x^\delta-1.$$
		The kernel of the map $\phi_\delta$ is given by
		$$ker(\phi_\delta)=\{(a(x)|0)\in\mathcal{C}|a(x)\in\mathcal{R_\gamma}\}.$$
		Consider a set $J=\{a(x)\in\mathcal{R_\gamma}|(a(x)|0)\in\ker(\phi_\delta)\}$ then, $J$ is an ideal in $\mathcal{R_\gamma}$ and therefore from Theorem \ref{th 2.1} $J=\langle g_1(x)+ua_1(x)\rangle,\text{ where }g_1(x),a_1(x)\in\mathbb F_q[x]\text{ and }a_1(x)|g_1(x)|x^\gamma-1$ and hence $ker(\phi_\delta)=\langle (g_1(x)+ua_1(x)|0)\rangle$. Since $(g_2(x)+ua_2(x))\in\phi_\delta(\mathcal{C})$, there exists $t(x)\in\mathcal{R}[x]$ such that $(t(x)|g_2(x)+ua_2(x))\in\mathcal{C}$. Let us take any arbitrary $(r(x),s(x))\in\mathcal{C}$ then, $\phi_\delta((r(x),s(x)))=s(x)\in\phi_\delta(\mathcal{C})$ and $s(x)=k_2(x)(g_2(x)+ua_2(x))$ for some $k_2(x)\in\mathcal{R}[x]$
		then,
		\begin{align*}
			(r(x)|s(x))-k_2(x)\star(l(x)|g_2(x)+ua_2(x))=&(r(x)-k_2(x)l(x)|s(x)-k_2(x)(g_2(x)+ua_2(x)))\\
			=&(r(x)-k_2(x)l(x)|0)\in\ker(\phi_\delta).
		\end{align*}This means
		\begin{align*}
			(r(x)-k_2(x)l(x)|0)=&k_1(x)\star(g_1(x)+ua_1(x)|0)\text{ for some }k_1(x)\in\mathcal{R}[x]\\
			r(x)-k_2(x)l(x)=&k_1(x)(g_1(x)+ua_1(x)) \text{ mod }x^\gamma-1\\
			r(x)=&k_1(x)(g_1(x)+ua_1(x))+k_2(x)l(x) \text{ mod }x^\gamma-1\\
		\end{align*}		
		\begin{align*}
			(r(x)|s(x))=&(k_1(x)(g_1(x)+ua_1(x))+k_2(x)l(x)|k_2(x)(g_2(x)+ua_2(x)))\\
			=&(k_1(x)(g_1(x)+ua_1(x))|0)+(k_2(x)l(x)|k_2(x)(g_2(x)+ua_2(x)))\\
			=&k_1(x)\star((g_1(x)+ua_1(x))|0)+k_2(x)\star(l(x)|g_2(x)+ua_2(x))
		\end{align*}
		for some $k_1(x),k_2(x)\in\mathcal{R}[x]$. Hence,
		$$\mathcal{C}=\langle(g_1(x)+ua_1(x)|0),(l(x)|g_2(x)+ua_2(x)).$$
	\end{proof}
	
		\begin{lemma}\label{th3.4}
		Let $\mathcal{C}$ be a double cyclic code of length $(\gamma,\delta)$ over $\mathcal{R}$, where $g_1(x),a_1(x),g_2(x),a_2(x)$ and $t(x)$ are as in Theorem \ref{th3.1} then, we may assume that $deg(t(x))<deg(g_1(x)+ua_1(x))$.
	\end{lemma}
	\begin{proof}
		Suppose that $deg(t(x))\geq deg(g_1(x)+ua_1(x))$. Let $j=deg(t(x))-deg(g_1(x)+ua_1(x))$ and $\mathbf{C}=\langle(g_1(x)+ua_1(x)|0),(t(x)-lx^j(g_1(x)+ua_1(x))|g_2(x)+ua_2(x))\rangle$, where $l$ is the leading coefficient of $t(x)$ then, $\mathbf{C}\subseteq\mathcal{C}$. However we also have
		\begin{align*}
			(t(x)|g_2(x)+ua_2(x))=(t(x)-lx^j(g_1(x)+ua_1(x))|g_2(x)+ua_2(x))+lx^j\star(g_1(x)+ua_1(x)|0).
		\end{align*}
		Thus $\mathcal{C}\subseteq\mathbf{C}$ hence, $\mathbf{C}=\mathcal{C}$. Therefore, the degree of $t(x)$ can be reduced in $\mathcal{C}$ so that we may assume $deg(t(x))<deg(g_1(x)+ua_1(x))$.
	\end{proof}
	
	\begin{lemma}\label{lem3.4}
		Let $\mathcal{C}$ be a double cyclic code of length $(\gamma,\delta)$ over $\mathcal{R}$, where $g_1(x),a_1(x),g_2(x),a_2(x)$ and $t(x)$ are as in theorem \ref{th3.1} then,
		\begin{itemize}
			\item[\rm (i)] $g_1(x)+ua_1(x)$ divides $\frac{x^\delta-1}{a_2(x)}t(x)$ in $\mathcal{R_\gamma}$.
			\item[\rm (ii)] $g_1(x)+ua_1(x)$ divides $u\frac{x^\delta-1}{g_2(x)}t(x)$ in $\mathcal{R_\gamma}$.
		\end{itemize}
	\end{lemma}
	
	\begin{lemma}
		Let $\mathcal{C}$ be a double cyclic code of length $(\gamma,\delta)$ over $\mathcal{R}$, where $g_1(x),a_1(x),g_2(x),a_2(x)$ and $t(x)$ are as in theorem \ref{th3.1}. If $g_1(x)+ua_1(x)$ is coprime with $\frac{x^\delta-1}{a_2(x)}$ then, $t(x)=0$.
	\end{lemma}
	
	\begin{theorem}
		Let $\mathcal{C}$ be a double cyclic code of length ($\gamma$,$\delta$) over $\mathcal{R}$, where $g_1(x),a_1(x),g_2(x),a_2(x)$ and $t(x)$ satisfy the conditions of Theorem \ref{th3.1} then, $\mathcal{C}$ is separable if and only if $t(x)=0$.
	\end{theorem}
	\begin{proof}
		Suppose that $t(x)=0$. Then $\mathcal{C}=\langle(g_1(x)+ua_1(x)|0),(0|g_2(x)+ua_2(x))\rangle$ hence $\phi_\gamma(\mathcal{C})=\langle g_1(x)+ua_1(x)\rangle$ and $\phi_\delta(\mathcal{C})=\langle g_2(x)+ua_2(x)\rangle$. Thus we observe that $\mathcal{C}=\mathcal{C_\gamma}\times\mathcal{C_\delta}$, where $\mathcal{C_\gamma}$ is a cyclic code of length $\gamma$ over $\mathcal{R}$ generated by $g_1(x)+ua_1(x)$ and $\mathcal{C_\delta}$ is a cyclic code of length $\delta$ over $\mathcal{R}$ generated by $(g_2(x)+ua_2(x))$. This means $\mathcal{C}$ is separable.
		\par Conversely, suppose that $\mathcal{C}$ is separable and $\mathit{C}=\langle(g_1(x)+ua_1(x)|0),(0|g_2(x)+ua_2(x))\rangle$. Since $\mathcal{C}$ is separable there exist $\mathcal{C_\gamma}$ and $\mathcal{C_\delta}$ such that $\mathcal{C}=\mathcal{C_\gamma}\times\mathcal{C_\delta}$. The code $\mathcal{C}$ can be viewed as $\phi_\gamma(\mathcal{C})\times\phi_\delta(\mathcal{C})$, where $\phi_\gamma(\mathcal{C})$ and $\phi_\delta(\mathcal{C})$ are images of $\mathcal{C_\gamma}$ and $\mathcal{C_\delta}$ respectively. Since $t(x)\in\phi_\gamma(\mathcal{C})$ therefore $(t(x)|0)\in\phi_\gamma(\mathcal{C})\times\phi_\delta(\mathcal{C})=\mathcal{C}$ hence $\mathit{C}\subseteq\mathcal{C}$. Also we have $(t(x)|0)=\xi_1(x)\star(g_1(x)+ua_1(x)|0)+\xi_2(x)\star(t(x)|g_2(x)+ua_2(x))$, where  $\xi_1(x),\xi_2(x)\in\mathcal{R}[x]$. Then
		\begin{equation}\label{eq3.1}
			t(x)=\xi_1(x)(g_1(x)+ua_1(x))+\xi_2(x)t(x)\text{ mod }x^\gamma-1
		\end{equation}
		\begin{equation}\label{eq3.2}
			0=\xi_2(x)(g_2(x)+ua_2(x))\text{ mod }x^\delta-1.
		\end{equation}
		From equation (\ref{eq3.2}) we have either
		$\xi_2(x)=0$ or $\xi_2(x)=u\frac{x^\delta-1}{g_2(x)}\nu_1(x)$ or $\xi_2(x)=\frac{x^\delta-1}{a_2(x)}\nu_2(x)$, where $\nu_1(x),\nu_2(x)\in\mathcal{R}[x]$.
		For case $\xi_2(x)=0$ we have from equation (\ref{eq3.1}) that $t(x)=\xi_1(x)(g_1(x)+ua_1(x))$ then
		\begin{align*}
			(t(x)|g_2(x)+ua_2(x))=&(\xi_1(x)(g_1(x)+ua_1(x))|g_2(x)+ua_2(x))\\
			=&\xi_1(x)\star(g_1(x)+ua_1(x)|0)+(0|g_2(x)+ua_2(x))\in\mathit{C}.
		\end{align*}
		When $\xi_2(x)=u\frac{x^\delta-1}{g_2(x)}\nu_1(x)$ we have from equation (\ref{eq3.1}) that $t(x)=\xi_1(x)(g_1(x)+ua_1(x))+u\frac{x^\delta-1}{g_2(x)}\nu_1(x)t(x)$ and from Lemma (\ref{lem3.4}) $u\frac{x^\delta-1}{g_2(x)}\nu_1(x)t(x)=d(x)(g_1(x)+ua_1(x))\text{ mod }x^\gamma-1$ for some $d(x)\in\mathcal{R}[x]$ then
		\begin{align*}
			(t(x)|g_2(x)+ua_2(x))=&((\xi_1(x)+d(x))(g_1(x)+ua_1(x))|g_2(x)+ua_2(x))\\
			=&(\xi_1(x)+d(x))\star(g_1(x)+ua_1(x)|0)+(0|g_2(x)+ua_2(x))\in\mathit{C}.
		\end{align*}
		Finally, for $\xi_2(x)=\frac{x^\delta-1}{a_2(x)}\nu_2(x)$ we have from equation (\ref{eq3.1}) that $t(x)=\xi_1(x)(g_1(x)+ua_1(x))+\frac{x^\delta-1}{a_2(x)}\nu_2(x)t(x)$ and from Lemma (\ref{lem3.4}) $\frac{x^\delta-1}{a_2(x)}\nu_2(x)t(x)=f(x)(g_1(x)+ua_1(x))\text{ mod }x^\gamma-1$ for some $f(x)\in\mathcal{R}[x]$ then
		\begin{align*}
			(t(x)|g_2(x)+ua_2(x))=&((\xi_1(x)+f(x))(g_1(x)+ua_1(x))|g_2(x)+ua_2(x))\\
			=&(\xi_1(x)+f(x))\star(g_1(x)+ua_1(x)|0)+(0|g_2(x)+ua_2(x))\in\mathit{C}.
		\end{align*}
		 Hence, $\mathcal{C}\subseteq\mathit{C}$ and therefore $\mathit{C}=\mathcal{C}$. We conclude that $\mathcal{C}=\langle(g_1(x)+ua_1(x)|0),(0|g_2(x)+ua_2(x))\rangle$.
	\end{proof}
	\begin{corollary}
		Let $\mathcal{C}$ be double cyclic code of length $(\gamma,\delta)$ over $\mathcal{R}$, where $g_1(x),a_1(x),g_2(x),a_2(x)$ and $t(x)$ satisfy the conditions in theorem \ref{th3.1} and $g_1(x)+ua_1(x)$ is coprime with $\frac{x^\delta-1}{a_2(x)}$ then, $\mathcal{C}$ is separable.
	\end{corollary}
	\section{Minimal spanning set}\label{s4}
	\begin{theorem}
		Let $\mathcal{C}$ be a double cyclic code of length $(\gamma,\delta)$ over $\mathcal{R}$, where $g_1(x),a_1(x),g_2(x),a_2(x)$ and $t(x)$ are as in Theorem \ref{th3.1} such that $x^\gamma-1=g_1(x)h_1(x)$, $x^\delta-1=g_2(x)h_2(x)$ with $deg(g_1(x))=r_1,\ deg(a_1(x))=r_2,\ deg(g_2(x))=s_1,\ deg(a_2(x))=s_2$. Consider the sets
		\begin{align*}
			N_1=&\bigcup\limits_{k=0}^{\gamma-r_1-1}\big\{x^k\star(g_1(x)+ua_1(x)|0)\big\}\\
			N_2=&\bigcup\limits_{k=0}^{r_1-r_2-1}\big\{x^k\star(uh_1(x)a_1(x)|0)\big\}\\
			N_3=&\bigcup\limits_{k=0}^{\delta-s_1-1}\big\{x^k\star(t(x)|g_2(x)+ua_2(x))\big\}\\
			N_4=&\bigcup\limits_{k=0}^{s_1-s_2-1}\big\{x^k\star(h_2(x)t(x)|uh_2(x)a_2(x))\big\}.
		\end{align*}
		Then $N_1\cup N_2\cup N_3\cup N_4$ forms a minimal generating set for $\mathcal{C}$ as an $\mathcal{R}$ module. Moreover, $\mathcal{C}$ has $q^{2\gamma+2\delta-r_1-s_1-r_2-s_2}$ codewords.
	\end{theorem}
	\begin{proof}
		Let $c(x)\in\mathcal{C}$ then there exist polynomials $p(x),q(x)\in\mathcal{R}[x]$ such that
		$$c(x)=p(x)\star(g_1(x)+ua_1(x)|0)+q(x)\star(t(x)|g_2(x)+ua_2(x)).$$
		If $deg(p(x))\leq\gamma-r_1-1$ then we have $p(x)\star(g_1(x)+ua_1(x)|0)\in span(N_1)$. Otherwise by division algorithm, we have
		$$p(x)=h_1(x)p_1(x)+l_1(x),$$ where $p_1(x),l_1(x)\in\mathcal{R}[x]$ and $l_1(x)=0$ or $deg(l_1(x))\leq\gamma-r_1-1$. Therefore, we have
		\begin{align*}
			p(x)\star(g_1(x)+ua_1(x)|0)=&(h_1(x)p_1(x)+l_1(x))\star(g_1(x)+ua_1(x)|0)\\
			=&p_1(x)\star(uh_1(x)a_1(x)|0)+l_1(x)\star(g_1(x)+ua_1(x)|0).
		\end{align*}
		If $deg(p_1(x))\leq r_1-r_2-1$ then, $p_1(x)\star(uh_1(x)a_1(x)|0)\in Span(N_2)$. Otherwise by division algorithm
		$$p_1(x)=\frac{x^\gamma-1}{h_1(x)a_1(x)}p_2(x)+l_2(x),$$ where $p_2(x),l_2(x)\in\mathcal{R}[x]$ and $l_2(x)=0$ or $deg(l_2(x))\leq r_1-r_2-1$. Therefore we have
		\begin{align*}
			p_1(x)\star(uh_1(x)a_1(x)|0)=&\big(\frac{x^\gamma-1}{h_1(x)a_1(x)}p_2(x)+l_2(x)\big)\star(uh_1(x)a_1(x)|0)\\
			=&l_2(x)\star(uh_1(x)a_1(x)|0)\in Span(N_2).
		\end{align*}
		Therefore, $p(x)\star(g_1(x)+ua_1(x)|0)\in Span(N_1\cup N_2)$.
		\par Now if $deg(q(x))\leq\delta-s_1-1$, then $q(x)\star(t(x)|g_2(x)+ua_2(x))\in Span(N_3)$. Otherwise from division algorithm $$q(x)=h_2(x)q_1(x)+l_3(x),$$ where $q_1(x),l_3(x)\in\mathcal{R}[x]$ and $l_3(x)=0$ or $deg(l_3(x))\leq\delta-s_1-1$. Therefore, we have
		\begin{align*}
			q(x)\star(t(x)|g_2(x)+ua_2(x))=&(h_2(x)q_1(x)+l_3(x))\star(t(x)|g_2(x)+ua_2(x))\\
			=&q_1(x)\star(h_2(x)t(x)|uh_2(x)a_2(x))+l_3(x)\star(t(x)|g_2(x)+ua_2(x)).
		\end{align*}
		If $deg(q_1(x))\leq s_1-s_2-1$, then $q_1(x)\star(h_2(x)t(x)|uh_2(x)a_2(x))\in Span(N_4)$. Otherwise using division algorithm we get
		$$q_1(x)=\frac{x^\delta-1}{h_2(x)a_2(x)}q_2(x)+l_4(x),$$ where $q_2(x),l_4(x)\in\mathcal{R}[x]$ and $l_4(x)=0$ or $deg(l_4(x))\leq s_1-s_2-1$. Therefore,
		\begin{align*}
			q_1(x)\star(h_2(x)t(x)|uh_2(x)a_2(x))=&\big(\frac{x^\delta-1}{h_2(x)a_2(x)}q_2(x)+l_4(x)\big)\star\big(h_2(x)t(x)|uh_2(x)a_2(x)\big)\\
			=&q_2(x)\star\big(\frac{x^\delta-1}{a_2(x)}t(x)|0\big)+l_4(x)\star\big(h_2(x)t(x)|uh_2(x)a_2(x)\big).
		\end{align*}
		From Lemma \ref{lem3.4}, $q_2(x)\star\big(\frac{x^\delta-1}{a_2(x)}t(x)|0\big)\in Span(N_1\cup N_2)$ and $l_4(x)\star\big(h_2(x)t(x)|uh_2(x)a_2(x)\big)\in Span(N_4).$ Therefore $N_1\cup N_2\cup N_3\cup N_4$ is a spanning set for $\mathcal{C}$. Since no element in $N_1\cup N_2\cup N_3\cup N_4$ is linearly dependent with other elements so, it is a minimal spanning set for $\mathcal{C}$. Clearly $\mathcal{C}$ has $q^{2\gamma+2\delta-r_1-s_1-r_2-s_2}$ codewords.
	\end{proof}
		Some Optimal codes (shown as $*$) according to online database Grassl \url{http://www.codetables.de/} are constructed in Table \ref{o1} and \ref{o2}, which are gray images of double cyclic codes over $\mathcal{R}$.
	\begin{table}[ht]
		\caption{Optimal binary code obtained from $\mathcal{R}$-double cyclic code}
		\centering 
		\begin{tabular}{c|c|c}
			\hline
			Generators of $\mathcal{C}$  & $(\gamma,\delta)$  & Parametre of $\Phi(\mathcal{C})$\\ \hline
			$g_1(x)=1+x+x^2,g_2(x)=1$     &   $(3,9)$ & $\ [24,20,2]^*$\\
			$g_1(x)=1+x+x^2,t(x)=1,g_2(x)=1+x^3+x^6+x^9+x^{12}$     &   $(3,15)$ & $[36,8,3]$\\
			$t(x)=u(1+x+x^3+x^4+x^6+x^7),a_2(x)=1+x+x^3+x^4+x^6+x^7$  &   $(9,9)$ &$\ [36,2,24]^*$\\
			$t(x)=u(1+x+x^2+x^3+x^4+x^5+x^6+x^7+x^8+x^9+x^{10}),$&&\\
			$a_2(x)=(1+x+x^2+x^3+x^4+x^5+x^6+x^7+x^8+x^9+x^{10})$               &  $(11,11)$ & $\ [44,1,44]^*$\\
			$t(x)=u(1+x+x^3+x^4+x^6+x^7+x^9+x^{10}+x^{12}+x^{13}),$&&\\$a_2(x)=1+x+x^3+x^4+x^6+x^7+x^9+x^{10}+x^{12}+x^{13}$               &  (15,15) & $\ [60,2,40]^*$\\
			\hline
			\end{tabular}
			\label{o1}
		\end{table}			
	\begin{table}[ht]
		\caption{Ternary code obtained from $\mathcal{R}$-double cyclic code}
		\centering 
		\small\begin{tabular}{c|c|c}
			\hline
			Generators of $\mathcal{C}$  & $(\gamma,\delta)$  & Parametre of\\ &&$\Phi(\mathcal{C})$\\ \hline
			$g_1(x)=1+x=a_1(x),t(x)=1=g_2(x)$ & $(2,2)$ & $[8,6,2]^*$\\ 
			$t(x)=u(1+x+x^2+x^3),a_2(x)=1+x+x^2+x^3$     &   $(4,4)$ & $[16,1,16]^*$\\
			$t(x)=u(2+x),g_2(x)=1+x$     &   $(4,4)$ & $[16,6,6]$\\
			$g_1(x)=1+x+x^2+x^3=a_1(x),t(x)=u,g_2(x)=1+x^4+x^8+x^{12},a_2(x)=1+x^2$     &   $(4,16)$ & $[16,20,4]$\\
			$t(x)=u(1+x+x^2+x^3+x^4+x^5+x^6+x^7+x^8+x^9),$&&\\
			$g_2(x)=1+x+x^2+x^3+x^4+x^5+x^6+x^7+x^8+x^9$     &   $(10,10)$ & $[40,2,30]^*$\\
			\hline
		\end{tabular}
		\label{o2}
	\end{table}
	\section{Dual}\label{s5}
	In this section, we find the relation between the generator polynomials of double cyclic code and its dual code. For a double cyclic code $\mathcal{C}=\langle(g_1(x)+ua_1(x)|0),(t(x)|g_2(x)+ua_2(x))\rangle$ with generators as in Theorem \ref{th3.1} of length $(\gamma,\delta)$ over $\mathcal{R}$, let us denote $g_1(x)+ua_1(x)$ by $F_1(x)$ and $g_2(x)+ua_2(x)$ by $G_1(x)$ for simplicity and we also assume $F_1(x)\text{ and }G_1(x)$ to be monic over $\mathcal{R}[x]$. Then $\mathcal{C}=\langle(F_1(x)|0),(t(x)|G_1(x))\rangle$.
	\par From \cite{j Gao}, we have that the dual $\mathcal{C}^\perp$ of $\mathcal{C}$ is also a double cyclic code of length $(\gamma,\delta)$ over $\mathcal{R}$, hence we denote $\mathcal{C}^\perp=\langle(\widehat{F_1}(x)|0),(\widehat{t}(x)|\widehat{G_1}(x))\rangle$.
	For any polynomial $f(x)\in\mathcal{R}[x]$ of degree $e$, the reciprocal polynomial $f^*(x)$ of $f(x)$ is defined as $f^*(x)=x^ef(x^{-1})$. The polynomial $f(x)$ is said to be self-reciprocal if $f^*(x)=f(x)$.
	\begin{lemma}\cite{borges}\label{l5.1}
		Let $m,n\in\mathbb{Z}^+$. Then $x^{mn}-1=(x^n-1)\theta_m(x^n)$, where $\theta_m(x)=\sum_{i=0}^{m-1}x^i$.
	\end{lemma}
	
	Let $\xi=lcm(\gamma,\delta)$, motivated by Borges et al., \cite{borges} we define the following map:
	$$\varphi:\mathcal{R_{\gamma,\delta}}\times\mathcal{R_{\gamma,\delta}}\to\frac{\mathcal{R}[x]}{\langle x^\xi-1\rangle},$$ such that for any two elements $\textbf{v}(x)=(v(x)|v'(x))$ and $\textbf{w}(x)=(w(x)|w'(x))$ of $\mathcal{R_{\gamma,\delta}}$, we have
	$$\varphi(\textbf{v}(x),\textbf{w}(x))=v(x)\theta_{\frac{\xi}{\gamma}}(x^\gamma)x^{\xi-1-deg(w(x))}w^*(x)+v'(x)\theta_{\frac{\xi}{\delta}}(x^\delta)x^{\xi-1-deg(w'(x))}w'^*(x).$$
	\par $\varphi$ is a bilinear map between $\mathcal{R}[x]$-modules.
	\begin{lemma}
		Let $\textbf{v}$ and $\textbf{w}$ be elements of $\mathcal{R^{\gamma}}\times\mathcal{R^{\delta}}$ with associated polynomials $\textbf{v}(x)=(v(x)|v'(x))$ and $\textbf{w}(x)=(w(x)|w'(x))$ respectively. Then $\textbf{v}$ is orthogonal to $\textbf{w}$ and all of its cyclic shifts if and only if $\varphi(\textbf{v}(x),\textbf{w}(x))=0$.
	\end{lemma}
	\begin{proof}
		For $\textbf{v}=(v_0,v_1,\ldots,v_{\gamma-1}|v'_0,v'_1,\ldots,v'_{\delta-1})$ and $\textbf{w}=(w_0,w_1,\ldots,v_{\gamma-1}|w'_0,w'_1,\ldots,w'_{\delta-1})$ in $\mathcal{R^{\gamma}}\times\mathcal{R^{\delta}}$, let $\textbf{v}^{(i)}=(v_{0+i},v_{1+i},\ldots,v_{\gamma-1+i}|v'_{0+i},v'_{1+i},\ldots,v'_{\delta-1+i})$ be the $i^{th}$ double cyclic shift of $\textbf{v}$. Then, $\textbf{w}\cdot\textbf{v}^{(i)}=0$ if and only if $\sum_{j=0}^{\gamma-1}w_jv_{j+i}+\sum_{d=0}^{\delta-1}w'_dv'_{d+i}=0.$ Let $S_i=\sum_{j=0}^{\gamma-1}w_jv_{j+i}+\sum_{d=0}^{\delta-1}w'_dv'_{d+i}$, then
		\begin{align*}
			\varphi(\textbf{v}(x),\textbf{w}(x))&=\sum_{e=0}^{\gamma-1}\big(\theta_{\frac{\xi}{\gamma}}(x^\gamma)\sum_{j=0}^{\gamma-1}w_jv_{j+e}x^{\xi-1-e}\big)+\sum_{l=0}^{\delta-1}\big(\theta_{\frac{\xi}{\delta}}(x^\delta)\sum_{d=0}^{\delta-1}w'_dv_{d+l}x^{\xi-1-l}\big)\\
			&=\sum_{i=0}^{\xi-1}S_ix^{\xi-1-i}
		\end{align*}
		in $\frac{\mathcal{R}[x]}{\langle x^\xi-1\rangle}$. Thus, $\varphi(\textbf{v}(x),\textbf{w}(x))=0$ if and only if $S_i=0$ for all $0\leq i\leq\xi-1$.
	\end{proof}
	\begin{lemma}\label{l5.3}
		Let $(F(x)|0)\text{ and }(l(x)|G(x))$ be the elements in $\mathcal{R_{\gamma,\delta}}$ such that $\varphi((F(x)|0),(l(x)|G(x)))=0$, then $F(x)l^*(x)=0$ in $\frac{\mathcal{R}[x]}{\langle x^\gamma-1\rangle}$. If $(0|D_1(x))\text{ and }(m(x)|D_2(x))$ be the elements in $\mathcal{R_{\gamma,\delta}}$ such that $\varphi((0|D_1(x)),(m(x)|D_2(x)))=0$, then $D_1(x)D_2^*(x)=0$ in $\frac{\mathcal{R}[x]}{\langle x^\delta-1\rangle}$.
	\end{lemma}
	\begin{proof}
		Since
		\begin{align*}
			\varphi((F(x)|0),(l(x)|G(x)))&=0\\
			F(x)\theta_{\frac{\xi}{\gamma}}(x^\gamma)x^{\xi-1-deg(l(x))}l^*(x)&=0,
		\end{align*}
		in $\frac{\mathcal{R}[x]}{\langle x^\xi-1\rangle}$. This means that $F(x)\theta_{\frac{\xi}{\gamma}}(x^\gamma)x^{\xi-1-deg(l(x))}l^*(x)=f(x)(x^\xi-1)$ for some $f(x)\in\mathcal{R}[x]$. Suppose that $g(x)=f(x)x^{deg(l(x))+1}$, then we have
		$$F(x)\theta_{\frac{\xi}{\gamma}}(x^\gamma)x^{\xi}l^*(x)=g(x)(x^\xi-1),$$ from Lemma \ref{l5.1} we get,
		$$F(x)x^{\xi}l^*(x)=g(x)(x^\gamma-1).$$ Since $x$ and $x^\gamma-1$ are co-prime, we have $F(x)l^*(x)=0$ in $\frac{\mathcal{R}[x]}{\langle x^\gamma-1\rangle}$.
		\par The other case can be proved by the similar argument. 
	\end{proof}
	\begin{proposition}
		Let $\mathcal{C}=\langle(g_1(x)+ua_1(x)|0),(0|g_2(x)+ua_2(x))\rangle$ be a double cyclic code of length ($\gamma$,$\delta$) over $\mathcal{R}$, where $g_1(x),a_1(x),g_2(x)$ and $a_2(x)$ satisfy the conditions of Theorem \ref{th3.1} then, $\mathcal{C}^\perp=\big\langle\big(\big(\frac{x^\gamma-1}{a_1(x)}\big)^*+u\big(\frac{x^\gamma-1}{g_1(x)}\big)^*|0\big),\big(0|\big(\frac{x^\delta-1}{a_2(x)}\big)^*+u\big(\frac{x^\delta-1}{g_2(x)}\big)^*\big)\big\rangle$.
	\end{proposition}
	\begin{proof}
		Let $\mathcal{C}$ be a separable double cyclic code of length $(\gamma,\delta)$ over $\mathcal{R}$ then, $\mathcal{C}=\mathcal{C_\gamma}\times\mathcal{C_\delta}$, where $\mathcal{C_\gamma}=\langle g_1(x)+ua_1(x)\rangle=\langle g_1(x),ua_1(x)\rangle$ and $\mathcal{C_\delta}=\langle g_2(x)+ua_2(x)\rangle=\langle g_2(x),ua_2(x)\rangle$. Then it is easy to see that $\mathcal{C}^\perp=\mathcal{C_\gamma}^\perp\times\mathcal{C_\delta}^\perp$. By \cite{rehman} we have $\mathcal{C_\gamma}^\perp=\big\langle\big(\frac{x^\gamma-1}{a_1(x)}\big)^*,u\big(\frac{x^\gamma-1}{g_1(x)}\big)^*\big\rangle$ and $\mathcal{C_\delta}^\perp=\big\langle\big(\frac{x^\delta-1}{a_2(x)}\big)^*,u\big(\frac{x^\delta-1}{g_2(x)}\big)^*\big\rangle$. Hence the result follows.
	\end{proof}
	
	\begin{proposition}
		Let $\mathcal{C}=\langle(F_1(x)|0),(t(x)|G_1(x))\rangle$ be a double cyclic code of length $(\gamma,\delta)$ over $\mathcal{R}$. Let $\mathcal{C}^\perp=\langle(\widehat{F_1}(x)|0),(\widehat{t}(x)|\widehat{G_1}(x))\rangle$ be its dual code. Then $\widehat{F_1}^*(x)gcd(F_1(x),t(x))=\pi(x)(x^\gamma-1)$, for some $\pi(x)\in\mathcal{R}[x]$.
	\end{proposition}
	\begin{proof}
		Since $(\widehat{F_1}(x)|0)\in\mathcal{C}^\perp$, then $\varphi((F_1(x)|0),(\widehat{F_1}(x)|0))=0$ and $\varphi((t(x)|G_1(x)),(\widehat{F_1}(x)|0))=0$ in $\frac{\mathcal{R}[x]}{\langle x^\xi-1\rangle}$. Hence from Lemma \ref{l5.3}, we have $F_1(x)\widehat{F_1}^*(x)=0$ and $t(x)\widehat{F_1}^*(x)=0$ in $\frac{\mathcal{R}[x]}{\langle x^\gamma-1\rangle}$. This means $gcd(F_1(x),t(x))\widehat{F_1}^*(x)=0$ in $\frac{\mathcal{R}[x]}{\langle x^\gamma-1\rangle}$. Thus, there exists some $\pi(x)\in\mathcal{R}[x]$ such that $\widehat{F_1}^*(x)gcd(F_1(x),t(x))=\pi(x)(x^\gamma-1)$.
	\end{proof}
	\begin{proposition}
		Let $\mathcal{C}=\langle(F_1(x)|0),(t(x)|G_1(x))\rangle$ be a double cyclic code of length $(\gamma,\delta)$ over $\mathcal{R}$. Let $\mathcal{C}^\perp=\langle(\widehat{F_1}(x)|0),(\widehat{t}(x)|\widehat{G_1}(x))\rangle$ be its dual code. Then $\widehat{G_1}^*(x)F_1(x)G_1(x)=\mu(x)(x^\delta-1)gcd(F_1(x),t(x))$, for some $\mu(x)\in\mathcal{R}[x]$. 
	\end{proposition}
	\begin{proof}
		Since $(\widehat{t}(x)|\widehat{G_1}(x))\in\mathcal{C}^\perp$. Consider the codeword $c(x)=\frac{F_1(x)}{gcd(F_1(x),t(x))}\star(t(x)|G_1(x))-\frac{t(x)}{gcd(F_1(x),t(x))}\star(F_1(x)|0)=(0|G_1(x)\frac{F_1(x)}{gcd(F_1(x),t(x))})$ then, $\varphi((0|G_1(x)\frac{F_1(x)}{gcd(F_1(x),t(x))}),(\widehat{t}(x)|\widehat{G_1}(x)))=0$ in $\frac{\mathcal{R}[x]}{\langle x^\xi-1\rangle}$. Hence from Lemma \ref{l5.3}, we have $G_1(x)\frac{F_1(x)}{gcd(F_1(x),t(x))}\widehat{G_1}^*(x)=0$ in $\frac{\mathcal{R}[x]}{\langle x^\delta-1\rangle}$. Hence there exists $\mu(x)\in\mathcal{R}[x]$ such that $G_1(x)F_1(x)\widehat{G_1}^*(x)=\mu(x)(x^\delta-1)gcd(F_1(x),t(x))$.
	\end{proof}
	
	\section{Reversible double cyclic codes}\label{s6}
	In this section we mainly focus on the reversibility of double cyclic code over $\mathcal{R}$. For any vector $c=(b|d)=(b_0,b_1,\ldots,b_{\gamma-1}|d_0,d_1,\ldots,d_{\delta-1})$ in $\mathcal{R^{\gamma}}\times\mathcal{R^{\delta}}$, the reverse $c^r$ of $c$ is defined as $c^r=(b_{\gamma-1},b_{\gamma-2},\ldots,b_0|d_{\delta-1},d_{\delta-2},\ldots,d_0)=(b^r|d^r)$.
	\begin{definition}
		A double cyclic code $\mathcal{C}$ of length $(\gamma,\delta)$ over $\mathcal{R}$ is said to be reversible double cyclic code if for all $c\in\mathcal{C}$, $c^r\in\mathcal{C}$.
	\end{definition}
	
	\begin{lemma}
		Let $f(x),g(x)$ be two polynomials in $\mathcal{R}[x]$ with $deg(g(x))\leq deg(f(x))$. Then,\\
		$(i)\ (f(x)g(x))^*=f^*(x)g^*(x)$,\\
		$(ii)\ (f(x)+g(x))^*=f^*(x)+x^{deg(f(x))-deg(g(x))}g^*(x)$.
	\end{lemma}
	The reversibility condition for a cyclic code over $\mathcal{R}$ with length coprime to $q$ is given as:
	\begin{theorem}\cite{prakash}\label{th6.2}
		Let $\mathcal{C}$ be a cyclic code of length $n$ over $\mathcal{R}$ as in Theorem \ref{th 2.1} $(1)$. Then $\mathcal{C}$ is reversible if and only if $g(x)$ and $a(x)$ both are self-reciprocal polynomials.
	\end{theorem}
	Now let us find reversibility conditions for separable double cyclic codes over $\mathcal{R}$.
	\begin{theorem}
		Let $\mathcal{C}=\langle (g_1(x)+ua_1(x)|0),(0|g_2(x)+ua_2(x))\rangle$ be a double cyclic separable code of length $(\gamma,\delta)$ over $\mathcal{R}$, where $g_1(x),a_1(x),g_2(x),a_2(x)$ and $t(x)$ are as in Theorem \ref{th3.1}. Then $\mathcal{C}$ is reversible double cyclic code if and only if $g_1(x),a_1(x),g_2(x)\text{ and }a_2(x)$ are self-reciprocal polynomials.
	\end{theorem}
	\begin{proof}
		Let $\mathcal{C}$ is reversible. Then the coordinate projections  $\mathcal{C_\gamma}\text{ and }\mathcal{C_\delta}$ of $\mathcal{C}$ are reversible cyclic codes of length $\gamma\text{ and }\delta$ respectively over $\mathcal{R}$. Therefore $\phi_\gamma(\mathcal{C})=\langle g_1(x)+ua_1(x)\rangle$ and $\phi_\delta(\mathcal{C})=\langle g_2(x)+ua_2(x)\rangle$ are reversible cyclic codes over $\mathcal{R}$ of respective length $\gamma$ and $\delta$. Hence by Theorem \ref{th6.2}, $g_1(x),a_1(x)$ and $g_2(x),a_2(x)$ are self-reciprocal polynomials.
		\par Conversely, suppose that $g_1(x),a_1(x),g_2(x)\text{ and }a_2(x)$ are self-reciprocal polynomials. Since $\mathcal{C}=\mathcal{C_\gamma}\times\mathcal{C_\delta}$, where $\mathcal{C_\gamma}$ is a cyclic code of length $\gamma$ over $\mathcal{R}$ generated by $g_1(x)+ua_1(x)$. Then, by Theorem  \ref{th6.2}, $\mathcal{C_\gamma}$ is a reversible cyclic code. Similar argument yields that $\mathcal{C_\delta}$ is a reversible cyclic code. Now take any $c=(b|d)\in\mathcal{C}=\mathcal{C_\gamma}\times\mathcal{C_\delta}$. Then $b\in\mathcal{C_\gamma}$ and $d\in\mathcal{C_\delta}$ therefore $b^r\in\mathcal{C_\gamma}$ and $d^r\in\mathcal{C_\delta}$ hence $c^r=(b^r|d^r)\in\mathcal{C_\gamma}\times\mathcal{C_\delta}=\mathcal{C}$.
	\end{proof}
	
	\begin{lemma}\label{lemma6.5}\cite{kanlaya}
		Let $v(x),w(x)\text{ and }(b(x)|d(x))\in\mathcal{R_{\gamma,\delta}}$ and $i\in\mathbb{Z}^+$. Then
		\begin{itemize}
			\item[\rm 1.]  $[v(x)+w(x)]^r=[v(x)]^r+[w(x)]^r$
			\item[\rm 2.]  $(x^ib(x)|x^id(x))^r=x^{(m
				+1)\gamma-1-deg(x^ib(x))}\star(b^*(x)|0)+x^{(n
				+1)\delta-1-deg(x^id(x))}\star(0|d^*(x))$.
		\end{itemize}
		Where $m,n$ are $0$ or the smallest positive integers such that
		\begin{align*}
			&m\gamma-deg(x^ib(x))+deg([x^ib(x)](\text{ mod }x^\gamma-1))\geq0\text{ and}\\
			&n\delta-deg(x^id(x))+deg([x^id(x)](\text{ mod }x^\delta-1))\geq0.
		\end{align*}
	\end{lemma}
	
	\begin{lemma}\label{lemma6.6}
		Let $(b(x)|d(x))\in\mathcal{R_{\gamma,\delta}}$ and $i\in\mathbb{Z}^+$. Suppose $\delta=(kp+1)\gamma$ and $deg(d(x))=kp\gamma+deg(b(x))$, where $k\in\mathbb{Z}^+\cup\{0\}$. Then,
		$$(x^ib(x)|x^id(x))^r=x^{(M(kp+1)+1)\gamma-1-deg(x^ib(x))}\star(b^*(x)|d^*(x)).$$ Where $M$ is $0$ or the smallest positive integer such that
		\begin{align*}
			&M\gamma-deg(x^ib(x))+deg([x^ib(x)](\text{ mod }x^\gamma-1))\geq0\text{ and}\\
			&M\delta-deg(x^id(x))+deg([x^id(x)](\text{ mod }x^\delta-1))\geq0.
		\end{align*}
	\end{lemma}
	\begin{proof}
		Let $\delta=(kp+1)\gamma$ and $deg(d(x))=kp\gamma+deg(b(x))$, where $k\in\mathbb{Z}^+\cup\{0\}$. From above Lemma \ref{lemma6.5}, we have
		$$(x^ib(x)|x^id(x))^r=x^{(m
			+1)\gamma-1-deg(x^ib(x))}\star(b^*(x)|0)+x^{(n
			+1)\delta-1-deg(x^id(x))}\star(0|d^*(x)).$$
			Where $m,n$ are $0$ or the smallest positive integers such that
			\begin{align*}
				&m\gamma-deg(x^ib(x))+deg([x^ib(x)](\text{ mod }x^\gamma-1))\geq0\text{ and}\\
				&n\delta-deg(x^id(x))+deg([x^id(x)](\text{ mod }x^\delta-1))\geq0.
			\end{align*}
			Let $M=max\{m,n\}$ and without loss of generality let $M=m$. Then $m=n+n'$ for some $n'\in\mathbb{Z}^+$. Since $deg(x^id(x))=kp\gamma+deg(x^ib(x))$, we have
			$$(M(kp+1)+1)\gamma-1-deg(x^ib(x))=(n+1)\delta-1-deg(x^id(x))+n'\delta.$$
			Hence,
			\begin{align*}
				x^{(M(kp+1)+1)\gamma-1-deg(x^ib(x))}\star(b^*(x)|d^*(x))&=x^{(m+1)\gamma-1-deg(x^ib(x))}x^{kpm\gamma}\star(b^*(x)|0)\\
				&+x^{(n+1)\delta-1-deg(x^id(x))}x^{n'\delta}\star(0|d^*(x))\\
				&=x^{(m+1)\gamma-1-deg(x^ib(x))}\star(b^*(x)|0)\\
				&+x^{(n+1)\delta-1-deg(x^id(x))}\star(0|d^*(x))\\
				&=(x^ib(x)|x^id(x))^r.
			\end{align*}
			Therefore,
			$$(x^ib(x)|x^id(x))^r=x^{(M(kp+1)+1)\gamma-1-deg(x^ib(x))}\star(b^*(x)|d^*(x)).$$ Where $M$ is $0$ or the smallest positive integer such that
			\begin{align*}
				&M\gamma-deg(x^ib(x))+deg([x^ib(x)](\text{ mod }x^\gamma-1))\geq0\text{ and}\\
				&M\delta-deg(x^id(x))+deg([x^id(x)](\text{ mod }x^\delta-1))\geq0.
			\end{align*}
	\end{proof}
	
	\begin{theorem}\label{th6.7}
		Let $\mathcal{C}=\langle (g_1(x)+ua_1(x)|0),(t(x)|g_2(x)+ua_2(x))\rangle$ be a double cyclic code of length $(\gamma,\delta)$ over $\mathcal{R}$, where $g_1(x),a_1(x),g_2(x),a_2(x)\text{ and }t(x)$ are as in Theorem \ref{th3.1}. If $\mathcal{C}$ is reversible then, $g_1(x),a_1(x),g_2(x)\text{ and }a_2(x)$ are self-reciprocal polynomials.
	\end{theorem}
	\begin{proof}
		Let $\mathcal{C}$ be reversible. Then $\phi_\delta(\mathcal{C})=\langle g_2(x)+ua_2(x)\rangle$ is reversible cyclic code of length $\delta$ over $\mathcal{R}$, hence by Theorem \ref{th6.2}, $g_2(x)\text{ and }a_2(x)$ are self-reciprocal. Since $J=\langle (g_1(x)+ua_1(x)|0)\rangle=ker(\phi_\delta)$ is an $\mathcal{R}[x]$-submodule of $\mathcal{R_{\gamma,\delta}}$ thus a double cyclic code of length $(\gamma,\delta)$ over $\mathcal{R}$. If $J$ is not reversible then there exist some $w(x)\in J$ such that $w(x)^r\in\mathcal{C}\backslash J$. Then $w(x)=\alpha(x)\star(g_1(x)+ua_1(x)|0)$ and $w(x)^r=(x^{\gamma-1-deg(B(x))}B^*(x)|0)$, where $B(x)=\alpha(x)(g_1(x)+ua_1(x))(\text{ mod }x^\gamma-1)$. Thus,
		\begin{align*}
			w(x)^r&=(x^{\gamma-1-deg(B(x))}B^*(x)|0)\\
			&=\alpha_1(x)\star(g_1(x)+ua_1(x)|0)+\alpha_2(x)\star(t(x)|g_2(x)+ua_2(x)),
		\end{align*}
		where $\alpha_1(x)\text{ and }0\neq\alpha_2(x)\in\mathcal{R}[x]$. Then we have
		\begin{equation}\label{eq6.1}
			x^{\gamma-1-deg(B(x))}B^*(x)=\alpha_1(x)(g_1(x)+ua_1(x))+\alpha_2(x)t(x)(\text{ mod }x^\gamma-1)
		\end{equation}
		\begin{equation}
			0=\alpha_2(x)(g_2(x)+ua_2(x))(\text{ mod }x^\delta-1).
		\end{equation}
		Since $\alpha_2(x)\neq0$, thus $\alpha_2(x)=u\frac{x^\delta-1}{g_2(x)}\lambda_1(x)$ or $\alpha_2(x)=\frac{x^\delta-1}{a_2(x)}\lambda_2(x)$ for some $\lambda_1(x),\lambda_2(x)\in\mathcal{R}[x]$. If $\alpha_2(x)=u\frac{x^\delta-1}{g_2(x)}\lambda_1(x)$ then from (\ref{eq6.1}) and Lemma \ref{lem3.4}, we have $x^{\gamma-1-deg(B(x))}B^*(x)=q_1(x)(g_1(x)+ua_1(x))(\text{ mod }x^\gamma-1)$, for some $q_1(x)\mathcal{R}[x]$ and hence $w(x)^r\in J$. Similarly if $\alpha_2(x)=\frac{x^\delta-1}{a_2(x)}\lambda_2(x)$ we get $w(x)^r\in J$. Hence $J$ is reversible double cyclic code and therefore $\phi_\gamma(J)=\langle g_1(x)+ua_1(x)\rangle$ is reversible cyclic code of length $\gamma$ over $\mathcal{R}$. Thus by Theorem \ref{th6.2}, $g_1(x)\text{ and }a_1(x)$ are self-reciprocal.
	\end{proof}
	\begin{theorem}\label{th6.8}
		Let $\mathcal{C}=\langle (g_1(x)+ua_1(x)|0),(t(x)|g_2(x)+ua_2(x))\rangle$ be a double cyclic code of length $(\gamma,\delta=(kp+1)\gamma)$ over $\mathcal{R}$, where $g_1(x),a_1(x),g_2(x),a_2(x)$ and $t(x)$ are as Theorem \ref{th3.1} and $k\in\mathbb{Z}^+\cup\{0\}$. Suppose that $deg(g_2(x)+ua_2(x))=kp\gamma+deg(t(x))$. If $\mathcal{C}$ is reversible, then $(g_1(x)+ua_1(x))|t^*(x)-t(x)\big(1+(p-1+x^j)uf(x)a_2(x)+(p-1+x^j)l(x)\frac{x^\delta-1}{g_2(x)}\big)$ in $\mathcal{R_{\gamma}}$, where $f(x),l(x)\in\mathbb F_q[x]$ and $f(x)g_2(x)+l(x)\frac{x^\delta-1}{g_2(x)}=1$ and $j=deg(g_2(x))-deg(a_2(x))$.
	\end{theorem}
	\begin{proof}
		Let $\mathcal{C}$ is reversible then, $g_1(x),a_1(x),g_2(x),a_2(x)$ are self-reciprocal polynomials. Since $(t(x)|g_2(x)+ua_2(x))\in\mathcal{C}$. Then
		$$(t(x)|g_2(x)+ua_2(x))^r=(x^{\gamma-1-deg(t(x))}t^*(x)|x^{\delta-1-deg(g_2(x)+ua_2(x))}(g_2(x)+ua_2(x))^*)\in\mathcal{C}.$$ This means 
		\begin{align*}
			x^{kp\gamma+1+deg(t(x))}\star&(x^{\gamma-1-deg(t(x))}t^*(x)|x^{\delta-1-deg(g_2(x)+ua_2(x))}(g_2(x)+ua_2(x))^*)\\
			&=(t^*(x)|(g_2(x)+ua_2(x))^*)\in\mathcal{C}.
		\end{align*}
	Hence
	\begin{align*}
		(t^*(x)|(g_2(x)+ua_2(x))^*)&=(t^*(x)|g_2(x)+ux^ja_2(x))\\
		&=q_1(x)\star(g_1(x)+ua_1(x)|0)+q_2(x)\star(t(x)|g_2(x)+ua_2(x)),
	\end{align*}
	where $q_1(x),q_2(x)\in\mathcal{R}[x]$ and $j=deg(g_2(x))-deg(a_2(x))$. Then we have
	\begin{equation}\label{eq6.3}
		t^*(x)=q_1(x)(g_1(x)+ua_1(x))+q_2(x)t(x)(\text{ mod }x^\gamma-1)
	\end{equation}
	\begin{equation}\label{eq6.4}
		g_2(x)+ux^ja_2(x)=q_2(x)(g_2(x)+ua_2(x))(\text{ mod }x^\delta-1).
	\end{equation}
	Notice that
	\begin{align*}
		ug_2(x)&=u(g_2(x)+ua_2(x))(\text{ mod }x^\delta-1)\text{ and}\\
		\frac{x^\delta-1}{g_2(x)}ua_2(x)&=\frac{x^\delta-1}{g_2(x)}(g_2(x)+ua_2(x))(\text{ mod }x^\delta-1)
	\end{align*}
	and $gcd(g_2(x),\frac{x^\delta-1}{g_2(x)})=1$. Therefore $f(x)g_2(x)+l(x)\frac{x^\delta-1}{g_2(x)}=1$ for some $f(x),l(x)\in\mathbb F_q[x]$. Also
	\begin{align*}
		ua_2(x)&=ua_2(x)\big(f(x)g_2(x)+l(x)\frac{x^\delta-1}{g_2(x)}\big)\\
		&=ua_2(x)f(x)g_2(x)+ua_2(x)l(x)\frac{x^\delta-1}{g_2(x)}.
	\end{align*}
	Hence, we have
	$$ua_2(x)=\big(uf(x)a_2(x)+l(x)\frac{x^\delta-1}{g_2(x)}\big)(g_2(x)+ua_2(x))(\text{ mod }x^\delta-1).$$
	Thus,
	$$ux^ja_2(x)=x^j\big(uf(x)a_2(x)+l(x)\frac{x^\delta-1}{g_2(x)}\big)(g_2(x)+ua_2(x))(\text{ mod }x^\delta-1)\text{ and }$$
	\begin{align*}
		g_2(x)&=g_2(x)+ua_2(x)+(p-1)ua_2(x)\\
		&=(g_2(x)+ua_2(x))+(p-1)\big[uf(x)a_2(x)+l(x)\frac{x^\delta-1}{g_2(x)} \big](g_2(x)+ua_2(x))(\text{ mod }x^\delta-1)\\
		&=(g_2(x)+ua_2(x))\big[1+(p-1)uf(x)a_2(x)+(p-1)l(x)\frac{x^\delta-1}{g_2(x)} \big](\text{ mod }x^\delta-1).
	\end{align*}
	This implies that
	$$g_2(x)+ux^ja_2(x)=(g_2(x)+ua_2(x))\big[1+(p-1+x^j)uf(x)a_2(x)+(p-1+x^j)l(x)\frac{x^\delta-1}{g_2(x)} \big](\text{ mod }x^\delta-1).$$ Then from (\ref{eq6.4}) we have,
	$$\big[q_2(x)-\big(1+(p-1+x^j)uf(x)a_2(x)+(p-1+x^j)l(x)\frac{x^\delta-1}{g_2(x)}\big)\big](g_2(x)+ua_2(x))=0(\text{ mod }x^\delta-1).$$
	This means that $\big[q_2(x)-\big(1+(p-1+x^j)uf(x)a_2(x)+(p-1+x^j)l(x)\frac{x^\delta-1}{g_2(x)}\big)\big]=0$ or $u\frac{x^\delta-1}{g_2(x)}\lambda_1(x)$ or $\frac{x^\delta-1}{a_2(x)}\lambda_2(x)$ for some $\lambda_1(x),\lambda_2(x)\in\mathcal{R}[x]$.
	\par If $q_2(x)-\big(1+(p-1+x^j)uf(x)a_2(x)+(p-1+x^j)l(x)\frac{x^\delta-1}{g_2(x)}\big)=0$ then, $$q_2(x)=\big(1+(p-1+x^j)uf(x)a_2(x)+(p-1+x^j)l(x)\frac{x^\delta-1}{g_2(x)}\big).$$ Then from (\ref{eq6.3}) we get,
	\begin{align*}
		t^*(x)=&q_1(x)(g_1(x)+ua_1(x))+\\
		&\big(1+(p-1+x^j)uf(x)a_2(x)+(p-1+x^j)l(x)\frac{x^\delta-1}{g_2(x)}\big)t(x)(\text{ mod }x^\gamma-1)\\
		t^*(x)-t(x)&\big(1+(p-1+x^j)uf(x)a_2(x)+(p-1+x^j)l(x)\frac{x^\delta-1}{g_2(x)}\big)=q_1(x)(g_1(x)+ua_1(x))(\text{ mod }x^\gamma-1).
	\end{align*}
	This means that 
	$$(g_1(x)+ua_1(x))|t^*(x)-t(x)\big(1+(p-1+x^j)uf(x)a_2(x)+(p-1+x^j)l(x)\frac{x^\delta-1}{g_2(x)}\big)\text{ in }\mathcal{R_{\gamma}}.$$
	\par If $q_2(x)-\big(1+(p-1+x^j)uf(x)a_2(x)+(p-1+x^j)l(x)\frac{x^\delta-1}{g_2(x)}\big)=u\frac{x^\delta-1}{g_2(x)}\lambda_1(x)$ then, $$q_2(x)t(x)=u\frac{x^\delta-1}{g_2(x)}t(x)\lambda_1(x)+\big(1+(p-1+x^j)uf(x)a_2(x)+(p-1+x^j)l(x)\frac{x^\delta-1}{g_2(x)}\big).$$ Then from (\ref{eq6.3}) and part $(ii)$ of Lemma \ref{lem3.4}, we have
	$$(g_1(x)+ua_1(x))|t^*(x)-t(x)\big(1+(p-1+x^j)uf(x)a_2(x)+(p-1+x^j)l(x)\frac{x^\delta-1}{g_2(x)}\big)\text{ in }\mathcal{R_{\gamma}}.$$
	\par Similarly, if $q_2(x)-\big(1+(p-1+x^j)uf(x)a_2(x)+(p-1+x^j)l(x)\frac{x^\delta-1}{g_2(x)}\big)=\frac{x^\delta-1}{a_2(x)}\lambda_2(x)$ then, $$q_2(x)t(x)=\frac{x^\delta-1}{a_2(x)}t(x)\lambda_2(x)+\big(1+(p-1+x^j)uf(x)a_2(x)+(p-1+x^j)l(x)\frac{x^\delta-1}{g_2(x)}\big).$$ Then from (\ref{eq6.3}) and part $(i)$ of Lemma \ref{lem3.4}, we have
	$$(g_1(x)+ua_1(x))|t^*(x)-t(x)\big(1+(p-1+x^j)uf(x)a_2(x)+(p-1+x^j)l(x)\frac{x^\delta-1}{g_2(x)}\big)\text{ in }\mathcal{R_{\gamma}}.$$
	\end{proof}
	\begin{theorem}\label{th6.9}
		Let $\mathcal{C}=\langle (g_1(x)+ua_1(x)|0),(t(x)|g_2(x)+ua_2(x))\rangle$ be a double cyclic code of length $(\gamma,\delta=(kp+1)\gamma)$ over $\mathcal{R}$, where $g_1(x),a_1(x),g_2(x),a_2(x)$ and $t(x)$ are as Theorem \ref{th3.1} and $k\in\mathbb{Z}^+\cup\{0\}$. Suppose that $deg(g_2(x)+ua_2(x))=kp\gamma+deg(t(x))$. Then $\mathcal{C}$ is reversible if and only if $g_1(x),a_1(x),g_2(x)\text{ and }a_2(x)$ are self-reciprocal polynomials and $(g_1(x)+ua_1(x))|t^*(x)-t(x)\big(1+(p-1+x^j)uf(x)a_2(x)+(p-1+x^j)l(x)\frac{x^\delta-1}{g_2(x)}\big)$ in $\mathcal{R_{\gamma}}$, where $f(x),l(x)\in\mathbb F_q[x]$ and $f(x)g_2(x)+l(x)\frac{x^\delta-1}{g_2(x)}=1$ and $j=deg(g_2(x))-deg(a_2(x))$.
	\end{theorem}
	
	\begin{proof}
		Let $\mathcal{C}$ is reversible then, the result follows from Theorem \ref{th6.7} and \ref{th6.8}.
		\par Conversely, suppose that $g_1(x),a_1(x),g_2(x)\text{ and }a_2(x)$ are self-reciprocal polynomials and $(g_1(x)+ua_1(x))|t^*(x)-t(x)\big(1+(p-1+x^j)uf(x)a_2(x)+(p-1+x^j)l(x)\frac{x^\delta-1}{g_2(x)}\big)$ in $\mathcal{R_{\gamma}}$, $f(x)g_2(x)+l(x)\frac{x^\delta-1}{g_2(x)}=1$ and $j=deg(g_2(x))-deg(a_2(x))$. Then
		$$	g_2(x)+ux^ja_2(x)=\big[1+(p-1+x^j)uf(x)a_2(x)+(p-1+x^j)l(x)\frac{x^\delta-1}{g_2(x)}\big](g_2(x)+ua_2(x))\text{ mod }x^\delta-1$$
		and
		$$t^*(x)=h(x)(g_1(x)+ua_1(x))+t(x)\big(1+(p-1+x^j)uf(x)a_2(x)+(p-1+x^j)l(x)\frac{x^\delta-1}{g_2(x)}\big)\text{ mod }x^\gamma-1,$$
		for some $h(x)\in\mathcal{R}[x]$. Then, we have
		\begin{align*}
			(t^*(x)|(g_2(x)+ua_2(x))^*)&=(t^*(x)|g_2(x)+ux^ja_2(x))\\
			&=(h(x)(g_1(x)+ua_1(x))+t(x)\big(1+(p-1+x^j)uf(x)a_2(x)+(p-1+x^j)\\&l(x)\frac{x^\delta-1}{g_2(x)}\big)|\big[1+(p-1+x^j)uf(x)a_2(x)+(p-1+x^j)l(x)\frac{x^\delta-1}{g_2(x)}\big](g_2(x)+ua_2(x)))\\
			&=h(x)\star(g_1(x)+ua_1(x)|0)+(t(x)|g_2(x)+ua_2(x))\\
			&+[(p-1+x^j)uf(x)a_2(x)+(p-1+x^j)l(x)\frac{x^\delta-1}{g_2(x)}]\star(t(x)|g_2(x)+ua_2(x)).
		\end{align*}
		This means that $(t^*(x)|(g_2(x)+ua_2(x))^*)\in\mathcal{C}.$ Also in Theorem \ref{th6.7} we proved that $ker(\phi_\delta)=\langle(g_1(x)+ua_1(x)|0)\rangle$ is reversible, hence we conclude that $((g_1(x)+ua_1(x))^*|0)\in\mathcal{C}.$ Now take any $w(x)\in\mathcal{C}$ then, $w(x)=l_1(x)\star(g_1(x)+ua_1(x)|0)+l_2(x)\star(t(x)|g_2(x)+ua_2(x))$ for some $l_1(x),l_2(x)\in\mathcal{R}[x]$ then,
		\begin{align*}
			w(x)^r&=(l_1(x)(g_1(x)+ua_1(x))+l_2(x)t(x)|l_2(x)(g_2(x)+ua_2(x)))^r\\
			&=(l_1(x)(g_1(x)+ua_1(x))|0)^r+(l_2(x)t(x)|l_2(x)(g_2(x)+ua_2(x)))^r.
		\end{align*}
		Now we have,
		\begin{align*}
			(l_1(x)(g_1(x)+ua_1(x))|0)^r&=\big(\sum_{i=0}^{e}c_ix^i(g_1(x)+ua_1(x))|0\big)^r,\text{ where }l_1(x)=\sum_{i=0}^{e}c_ix^i\\
			&=\sum_{i=0}^ec_i\big(x^i(g_1(x)+ua_1(x))|0\big)^r.
		\end{align*}
		From Lemma \ref{lemma6.5} there exists $m_i\in\mathbb{Z}^+\cup\{0\}$ such that
		$$(x^i(g_1(x)+ua_1(x))|0)^r=x^{(m_i+1)\gamma-1-deg(x^i(g_1(x)+ua_1(x)))}\star((g_1(x)+ua_1(x))^*|0).$$ This means that $(l_1(x)(g_1(x)+ua_1(x))|0)^r\in\mathcal{C}$, and
		\begin{align*}
			(l_2(x)t(x)|l_2(x)(g_2(x)+ua_2(x)))^r&=\big(\sum_{i=0}^{o}s_ix^it(x))|\sum_{i=0}^{o}s_ix^i(g_2(x+ua_2(x)))\big)^r,\text{ where }l_2(x)=\sum_{i=0}^{o}s_ix^i\\
			&=\sum_{i=0}^{o}s_i\big(x^it(x))|x^i(g_2(x+ua_2(x)))\big)^r.
		\end{align*}
		By Lemma \ref{lemma6.6}, there exists $M_i\in\mathbb{Z}^+\cup\{0\}$ such that
		$$\big(x^it(x))|x^i(g_2(x+ua_2(x)))\big)^r=x^{(M_i(kp+1)+1)\gamma-1-deg(x^it(x)}\star(t^*(x)|(g_2(x)+ua_2(x))^*).$$ This implies that $(l_2(x)t(x)|l_2(x)(g_2(x)+ua_2(x)))^r\in\mathcal{C}$, this means $	w(x)^r\in\mathcal{C}$. Hence $\mathcal{C}$ is reversible.
	\end{proof}
	Notice that if we take $deg(g_2(x))=deg(a_2(x))$ then $j=0$ in the above Theorem \ref{th6.9}, then $(g_1(x)+ua_1(x))|t^*(x)-t(x)$, since $deg(t^*(x))\leq deg(t(x))\leq deg(g_1(x)+ua_1(x))$, this means that $t^*(x)-t(x)=0\implies t^*(x)=t(x)$ i.e., $t(x)$ is self-reciprocal polynomial. Hence we have the following corollary.
	\begin{corollary}
		Let $\mathcal{C}=\langle (g_1(x)+ua_1(x)|0),(t(x)|g_2(x)+ua_2(x))\rangle$ be a double cyclic code of length $(\gamma,\delta=(kp+1)\gamma)$ over $\mathcal{R}$, where $g_1(x),a_1(x),g_2(x),a_2(x)$ and $t(x)$ are as Theorem \ref{th3.1}, and $k\in\mathbb{Z}^+\cup\{0\}$. Suppose that $deg(g_2(x))=deg(a_2(x))=kp\gamma+deg(t(x))$. Then $\mathcal{C}$ is reversible if and only if $g_1(x),a_1(x),g_2(x),a_2(x)\text{ and } t(x)$ are self-reciprocal polynomials.
	\end{corollary}

	\section{Reversible-complement double cyclic codes and DNA codes}\label{s7}
	Througout this section $\mathcal{R}$ represents the ring $\mathcal{R}=\mathbb F_4+u\mathbb F_4$, $u^2=0$, where $\mathbb F_4=\{0,1,\alpha,\alpha^2\}$ is a finite field of four elements. In this section studied reversible-complement double cyclic codes over $\mathcal{R}$ and constructed some DNA codes derived from double cyclic DNA codes.
	\par Deoxyribonucleic acid or DNA is an acid found in almost every living organism mostly found in the nucleus of the cell (Eukaryotic cell) that contains the genetic information of the living organism. It is a sequence of two long polymers, called strands which are composed of four nucleotide bases namely Adenine(A), Guanine(G), Thymine(T) and Cytosine(C). Two strands are so twisted, forming a double helix, running in opposite directions to each other and joined together by hydrogen bonds between nucleotide bases. This attachment follows the Watson-Crick Complement rule. $A$ pairs with $T$ and $G$ pairs with $C$, as per the Watson-Crick Complement rule. $A$ and $G$ are called the complements of $T$ and $C$, respectively, and vice versa. The Complement of a base $X$ is denoted by $\bar{X}$. $\bar{G}=C$, for instance, is the complement of $G$. Thus, if $X=AGATT$ is a DNA strand, then $\bar{X}=TCTAA$ would be its complement. According to the Watson-Crick Complement rule, a DNA strand $Y=y_1y_2\ldots y_{l}$ will pair up with $Y^{rc}=\bar{y_{l}}\bar{y_{l-1}}\ldots\bar{y_2}\bar{y_1}$, the reverse-complement of $Y$. For instance, a DNA strand $5'-TCTAAGT-3'$ will pair up with $3'-ACTTAGA-5'$. A DNA code $\mathcal{C}$ with minimum distance $d$ may satisfy some or all the following constraints:
	\begin{itemize}
		\item [\rm (a) ] \textbf{The Hamming constraint:} $d_H(s_1,s_2)\geq d$, where $s_1,s_2\in \mathcal{C}$ and $s_1\neq s_2$.
		\item [\rm (b) ] \textbf{The Reverse constraint:} $d_H(s_1,s_2^r)\geq d$ including $s_1=s_2$, where $s_1,s_2\in \mathcal{C}$ and $s_2^r$ is the reverse of $s_2$.
		\item [\rm (c) ] \textbf{The Reverse-Complement constraint:} $d_H(s_1^r,s_2^c)\geq d$ including $s_1=s_2$, where $s_1,s_2\in \mathcal{C}$ and $s_2^c$ is the complement of $s_2$.
		\item [\rm (d) ] \textbf{The $GC$-content constraint:} Each codeword $s\in\mathcal{C}$ has the same number of $G$ or $C$.
	\end{itemize}
	First three constraints ensure to reduce the probability of non-specific hybridization. Fixed $GC$-content constraint ensures the similar melting point. In \cite{sri} authors provided a one-to-one correspondence $\boldsymbol{\tau}$ between the elements of the ring $\mathcal{R}$ and the set of all DNA double base pairs $S_{D_{16}}$ which we provided in Table \ref{ta1}. Notice that in Table \ref{ta1} the mapping $\boldsymbol{\tau}$ preserves the complementary property and when an element of $\mathcal{R}$ is multiplied by $(1+u)$, the corresponding DNA pair is reversed. Therefore, if $w_1w_2\cdots w_{2\gamma}$ is a DNA sequence of $\gamma$-tuple $w$ in $\mathcal{R}^\gamma$ then, the DNA sequence of $(1+u)w^r$ is $w_{2\gamma} w_{2\gamma-1}\cdots w_1$ and hence we conclude that,
	\begin{lemma}
		Let $c=(b_0,\ldots,b_{\gamma-1}|d_0,\ldots,d_{\delta-1})$ be any codeword of a double cyclic code $\mathcal{C}$ of length $(\gamma,\delta)$ over $\mathcal{R}$ and $X=x_1x_2\cdots x_{2\gamma}y_1y_2\cdots y_{2\delta}$ is the DNA sequence corresponding to $c$. Then, the DNA sequence corresponding to the codeword $(1+u)c^r$ is $x_{2\gamma}x_{2\gamma-1}\cdots x_1y_{2\delta}y_{2\delta-1}\cdots y_1$.
	\end{lemma}
	For any vector $w=(b|d)=(b_0,b_1,\ldots,b_{\gamma-1}|d_0,d_1,\ldots,d_{\delta-1})\in\mathcal{R^{\gamma,\delta}}$, the complement of $w$ is defined as $w^c=(\overline{b_0},\overline{b_1},\ldots,\overline{b_{\gamma-1}}|\overline{d_0},\overline{d_1},\ldots,\overline{d_{\delta-1}})$ i.e., $w^c=(b^c|d^c)$ and the reverse-complement is defined as $w^{rc}=(\overline{b_{\gamma-1}},\overline{b_{\gamma-2}},\ldots,\overline{b_0}|\overline{d_{\delta-1}},\overline{d_{\delta-2}},\ldots,\overline{d_0})$ i.e., $w^{rc}=(b^{rc}|d^{rc})$.
	\begin{definition}
		A double cyclic code $\mathcal{C}$ of length $(\gamma,\delta)$ over $\mathcal{R}$ is said to be a double cyclic DNA code of length $(\gamma,\delta)$ over $\mathcal{R}$ if for any $w\in\mathcal{C},w\neq w^{rc}$, $w^{rc}\in\mathcal{C}$.
	\end{definition}
	\begin{theorem}\label{t7.1}
		Let $\mathcal{C}$ be a double cyclic code of length $(\gamma,\delta)$ over $\mathcal{R}$. Then $\mathcal{C}$ is reversible-complement double cyclic code if and only if
		\begin{itemize}
			\item [\rm (a). ] $\mathcal{C}$ is reversible double cyclic code and
			\item [\rm (b). ] $u\mathbf{I}=(\overbrace{u,\ldots,u}^\gamma|\overbrace{u,\ldots,u}^\delta)\in\mathcal{C}$, where $\mathbf{I}$ is all one vector i.e., $\mathbf{I}=(\overbrace{1,\ldots,1}^\gamma|\overbrace{1,\ldots,1}^\delta)$.
		\end{itemize}
	\end{theorem}
	\begin{proof}
		Let $\mathcal{C}$ be a reversible-complement code. Since $(\overbrace{0,\ldots,0}^\gamma|\overbrace{0,\ldots,0}^\delta)\in\mathcal{C}$ then we have that $(\overbrace{0,\ldots,0}^\gamma|\overbrace{0,\ldots,0}^\delta)^{rc}\in\mathcal{C}$. Then,
		\begin{align*}
			(\overbrace{0,\ldots,0}^\gamma|\overbrace{0,\ldots,0}^\delta)^{rc}&=(\overbrace{\overline{0},\ldots,\overline{0}}^\gamma|\overbrace{\overline{0},\ldots,\overline{0}}^\delta)\\
			&=(\overbrace{0+u,\ldots,0+u}^\gamma|\overbrace{0+u,\ldots,0+u}^\delta)\\
			&=(\overbrace{0,\ldots,0}^\gamma|\overbrace{0,\ldots,0}^\delta)+(\overbrace{u,\ldots,u}^\gamma|\overbrace{u,\ldots,u}^\delta).
		\end{align*}
		Since $(\overbrace{0,\ldots,0}^\gamma|\overbrace{0,\ldots,0}^\delta)^{rc}\in\mathcal{C}$ and $(\overbrace{0,\ldots,0}^\gamma|\overbrace{0,\ldots,0}^\delta)\in\mathcal{C}$, this means that $(\overbrace{u,\ldots,u}^\gamma|\overbrace{u,\ldots,u}^\delta)=u\mathbf{I}\in\mathcal{C}$. Let any $w=(b_0,b_1,\ldots,b_{\gamma-1}|d_0,d_1\ldots,d_{\delta-1})\in\mathcal{C}$. Then $w^{rc}\in\mathcal{C}$. Then,
		\begin{align*}
			w^{rc}&=(\overline{b_{\gamma-1}},\overline{b_{\gamma-2}},\ldots,\overline{b_0}|\overline{d_{\delta-1}},\overline{d_{\delta-2}},\ldots,\overline{d_0})\\
			&=(b_{\gamma-1}+u,b_{\gamma-2}+u,\ldots,b_0+u|d_{\delta-1}+u,d_{\delta-2}+u,\ldots,d_0+u)\\
			&=(b_{\gamma-1},b_{\gamma-2},\ldots,b_0|d_{\delta-1},d_{\delta-2},\ldots,d_0)+(\overbrace{u,\ldots,u}^\gamma|\overbrace{u,\ldots,u}^\delta)\\
			&=w^r+u\mathbf{I}.
		\end{align*}
		Since $w^{rc},u\mathbf{I}\in\mathcal{C}$, this means $w^r\in\mathcal{C}$, hence $\mathcal{C}$ is reversible.
		\par Conversely, suppose that $\mathcal{C}$ is reversible and $u\mathbf{I}\in\mathcal{C}$, then for any $w\in\mathcal{C}$ implies $w^r\in\mathcal{C}$ and we see that $w^{rc}=w^r+u\mathbf{I}$. This means for any $w\in\mathcal{C}$, $w^{rc}\in\mathcal{C}$, hence $\mathcal{C}$ reversible-complement.
	\end{proof}
	With the help of Theorem \ref{t7.1}, and Section \ref{s6}, the following theorems holds directly and we may consider $\gamma,\delta$ to be positive odd integers. 
	\begin{theorem}\label{t7.4}
		Let $\mathcal{C}=\langle(g_1(x)+ua_1(x)|0),(0|g_2(x)+ua_2(x))\rangle$ be a separable double cyclic code of length $(\gamma,\delta)$ over $\mathcal{R}$, where $g_1(x),a_1(x),g_2(x)\text{ and }a_2(x)$ satisfy conditions in Theorem \ref{th3.1}. Then $\mathcal{C}$ is reversible-complement double cyclic code if and only if $g_1(x),a_1(x),g_2(x)\text{ and }a_2(x)$ are self-reciprocal polynomials and $u\mathbf{I}\in\mathcal{C}$.
	\end{theorem}
	For non separable codes we have the following results.
	\begin{theorem}\label{t7.5}
		Let $\mathcal{C}=\langle (g_1(x)+ua_1(x)|0),(t(x)|g_2(x)+ua_2(x))\rangle$ be a double cyclic code of length $(\gamma,\delta=(kp+1)\gamma)$ over $\mathcal{R}$, where $g_1(x),a_1(x),g_2(x),a_2(x)$ and $t(x)$ are as Theorem \ref{th3.1} and $k\in\mathbb{Z}^+\cup\{0\}$. Suppose that $deg(g_2(x)+ua_2(x))=kp\gamma+deg(t(x))$. Then $\mathcal{C}$ is reversible-complement if and only if
		\begin{itemize}
			\item [\rm (i)] $g_1(x),a_1(x),g_2(x)\text{ and }a_2(x)$ are self-reciprocal polynomials,
			\item [\rm (ii)] $(g_1(x)+ua_1(x))|t^*(x)-t(x)\big(1+(p-1+x^j)uf(x)a_2(x)+(p-1+x^j)l(x)\frac{x^\delta-1}{g_2(x)}\big)$ in $\mathcal{R_{\gamma}}$, where $f(x),l(x)\in\mathbb F_q[x]$ and $f(x)g_2(x)+l(x)\frac{x^\delta-1}{g_2(x)}=1$ and $j=deg(g_2(x))-deg(a_2(x))$ and
			\item [\rm (iii)] $u\mathbf{I}\in\mathcal{C}$.
		\end{itemize}
	\end{theorem}
	
	\begin{corollary}
		Let $\mathcal{C}=\langle (g_1(x)+ua_1(x)|0),(t(x)|g_2(x)+ua_2(x))\rangle$ be a double cyclic code of length $(\gamma,\delta=(kp+1)\gamma)$ over $\mathcal{R}$, where $g_1(x),a_1(x),g_2(x),a_2(x)$ and $t(x)$ are as Theorem \ref{th3.1} and $k\in\mathbb{Z}^+\cup\{0\}$. Suppose that $deg(g_2(x))=deg(a_2(x))=kp\gamma+deg(t(x))$. Then $\mathcal{C}$ is reversible-complement if and only if $g_1(x),a_1(x),g_2(x),a_2(x)\text{ and } t(x)$ are self-reciprocal polynomials and $u\mathbf{I}\in\mathcal{C}$.
	\end{corollary}
	Let $\mathcal{C}$ be a double cyclic code of length $(\gamma,\delta)$ over $\mathcal{R}$ and $S_{D_{16}}$ is the set of all DNA double pairs, Table \ref{ta1} shows the correspondence $\tau$ between $\mathcal{R}$ and $S_{D_{16}}$. Now consider a correspondence
	\begin{align*}
		\Theta:\mathcal{C}\to & S_{D_{16}}^{\gamma+\delta},\text{ defined by}\\
		\Theta(b_0,b_1,\ldots,b_{\gamma-1}|d_0,d_1,\ldots,d_{\delta-1})&=(\tau(b_0),\tau(b_1),\ldots,\tau(b_{\gamma-1}),\tau(d_0),\tau(d_1),\ldots,\tau(d_{\delta-1})),
	\end{align*}
	$(b_0,b_1,\ldots,b_{\gamma-1}|d_0,d_1,\ldots,d_{\delta-1})\in\mathcal{C}$.For any double cyclic code of length $(\gamma,\delta)$ over $\mathcal{R}$, consider a set
	$$\mathcal{B}=\{(d|b)\in\mathcal{R^{\delta,\gamma}}|(b|d)\in\mathcal{C}\}\subseteq\mathcal{R^{\delta,\gamma}}.$$
	Then, a DNA code $\mathfrak{D}$ is generated from $\mathcal{C}\cup\mathcal{B}$, where $\mathcal{C}$ is a reversible-complement double cyclic DNA code of length $(\gamma,\delta)$ over $\mathcal{R}$.
	\begin{example}
		Let $\mathcal{C}=\langle(u(x^2+x+1)|0),(0|u(x^2+\alpha x+1)(x^2+\alpha^2 x+1))\rangle$ be a double cyclic code of length (3,5) over $\mathcal{R}$. Since $u\mathbf{I}=(u(x^2+x+1)|0)+(0|u(x^2+\alpha x+1)(x^2+\alpha^2 x+1))\in\mathcal{C}$, hence by Theorem \ref{t7.4}, $\mathcal{C}$ is a double cyclic DNA code of length $(3,5)$. Then a DNA code $\mathfrak{D}$ of length 16 with minimum distance 4 obtained from $\mathcal{C}\cup\mathcal{B}$ has 14 codewords, given in Table \ref{ta2}.
	
	\begin{table}[ht]
		\caption{DNA code $\mathfrak{D}$ of length 16 obtained from $\mathcal{C}\cup\mathcal{B}$}
		\centering
		\begin{tabular}{|c|c|}
			\hline
			$AAAAAAAAAAAAAAAA$  &  $TTTTTTAAAAAAAAAA$\\
			$AAAAAAGGGGGGGGGG$  &  $GGGGGGGGGGAAAAAA$\\
			$AAAAAATTTTTTTTTT$  &  $TTTTTTTTTTAAAAAA$\\
			$AAAAAACCCCCCCCCC$  &  $CCCCCCCCCCAAAAAA$\\
			$TTTTTTGGGGGGGGGG$  &  $AAAAAAAAAATTTTTT$\\
			$TTTTTTTTTTTTTTTT$  &  $GGGGGGGGGGTTTTTT$\\
			$TTTTTTCCCCCCCCCC$  &  $CCCCCCCCCCTTTTTT$\\
			\hline
		\end{tabular}
		\label{ta2}
	\end{table}
\end{example}
	\begin{example}
		Let $\mathcal{C}=\langle(u(x^2+x+1)|0),(x^2+x+1|x^8+x^7+x^6+x^5+x^4+x^3+x^2+x+1)\rangle$ be a double cyclic code of length (3,9) over $\mathcal{R}$. Since $u\mathbf{I}=u(x^2+x+1|x^8+x^7+x^6+x^5+x^4+x^3+x^2+x+1)\in\mathcal{C}$, hence by Theorem \ref{t7.5}, $\mathcal{C}$ is a double cyclic DNA code of length $(3,9)$. Then a DNA code $\mathfrak{D}$ of length 24 with minimum distance 6 obtained from $\mathcal{C}\cup\mathcal{B}$ has 48 codewords, given in Table \ref{ta3}.	
		\begin{table}[ht]
		\caption{DNA code $\mathfrak{D}$ of length 24 obtained from $\mathcal{C}\cup\mathcal{B}$}
		\centering
		\begin{tabular}{|c|c|}
			\hline
			$GCGCGCGCGCGCGCGCGCGCGCGC$  &  $CGCGCGCGCGCGCGCGCGCGCGCG$\\
			$GAGAGAGAGAGAGAGAGAGAGAGA$  &  $AGAGAGAGAGAGAGAGAGAGAGAG$\\
			$GTGTGTGTGTGTGTGTGTGTGTGT$  &  $TGTGTGTGTGTGTGTGTGTGTGTG$\\
			$GGGGGGGGGGGGGGGGGGGGGGGG$  &  $AAAAAAAAAAAAAAAAAAAAAAAA$\\
			$TATATATATATATATATATATATA$  &  $ATATATATATATATATATATATAT$\\
			\hline
			\end{tabular}
			\label{ta3}
		\end{table}
	
\begin{flushleft}
Table 5 continued.
\end{flushleft}
\begin{center}
\begin{tabular}{|c|c|}
\hline
			$TCTCTCTCTCTCTCTCTCTCTCTC$  &  $CTCTCTCTCTCTCTCTCTCTCTCT$\\
			$CACACACACACACACACACACACA$  &  $ACACACACACACACACACACACAC$\\
			$TTTTTTTTTTTTTTTTTTTTTTTT$  &  $CCCCCCCCCCCCCCCCCCCCCCCC$\\
			$CTCTCTGAGAGAGAGAGAGAGAGA$  &  $AGAGAGAGAGAGAGAGAGTCTCTC$\\
			$CGCGCGGCGCGCGCGCGCGCGCGC$  &  $CGCGCGCGCGCGCGCGCGGCGCGC$\\
			$CACACAGTGTGTGTGTGTGTGTGT$  &  $TGTGTGTGTGTGTGTGTGACACAC$\\
			$CCCCCCGGGGGGGGGGGGGGGGGG$  &  $GGGGGGGGGGGGGGGGGGCCCCCC$\\
			$ATATATTATATATATATATATATA$  &  $ATATATATATATATATATTATATA$\\
			$AGAGAGTCTCTCTCTCTCTCTCTC$  &  $CTCTCTCTCTCTCTCTCTGAGAGA$\\
			$ACACACTGTGTGTGTGTGTGTGTG$  &  $GTGTGTGTGTGTGTGTGTCACACA$\\
			$TCTCTCAGAGAGAGAGAGAGAGAG$  &  $GAGAGAGAGAGAGAGAGACTCTCT$\\
			$GCGCGCCGCGCGCGCGCGCGCGCG$  &  $GCGCGCGCGCGCGCGCGCCGCGCG$\\
			$GAGAGACTCTCTCTCTCTCTCTCT$  &  $TCTCTCTCTCTCTCTCTCAGAGAG$\\
			$GTGTGTCACACACACACACACACA$  &  $ACACACACACACACACACTGTGTG$\\
			$TTTTTTAAAAAAAAAAAAAAAAAA$  &  $AAAAAAAAAAAAAAAAAATTTTTT$\\
			$AAAAAATTTTTTTTTTTTTTTTTT$  &  $TTTTTTTTTTTTTTTTTTAAAAAA$\\
			$GGGGGGCCCCCCCCCCCCCCCCCC$  &  $CCCCCCCCCCCCCCCCCCGGGGGG$\\
			$TGTGTGACACACACACACACACAC$  &  $CACACACACACACACACAGTGTGT$\\
			$TATATAATATATATATATATATAT$  &  $TATATATATATATATATAATATAT$\\		
		\hline
		\end{tabular}
	\end{center}
\end{example}

	\begin{example}
		Let $\mathcal{C}=\langle(u(x^2+x+1)|0),(u|u(x^6+x^3+1))\rangle$ be a double cyclic code of length (3,9) over $\mathcal{R}$. Since $u\mathbf{I}=u(x^2+x+1|x^8+x^7+x^6+x^5+x^4+x^3+x^2+x+1)\in\mathcal{C}$, hence by Theorem \ref{t7.5}, $\mathcal{C}$ is a double cyclic DNA code of length $(3,9)$. Then a DNA code $\mathfrak{D}$ of length 24 with minimum distance 6 obtained from $\mathcal{C}\cup\mathcal{B}$, given in Table \ref{ta4}.
	\end{example}
	
	\begin{table}[ht]
		\caption{DNA code $\mathfrak{D}$ of length 24 obtained from $\mathcal{C}\cup\mathcal{B}$}
		\centering
		\begin{tabular}{|c|c|}
			\hline
			$GGAAAAGGAAAAGGAAAAGGAAAA$   &   $AAAAGGAAAAGGAAAAGGAAAAGG$\\
			$AAGGAAAAGGAAAAGGAAAAGGAA$   &   $GGAAGGGGAAGGGGAAGGGGAAGG$\\
			$GGGGAAGGGGAAGGGGAAGGGGAA$   &   $AAGGGGAAGGGGAAGGGGAAGGGG$\\
			$TTGGAATTGGAATTGGAATTGGAA$   &   $AAGGTTAAGGTTAAGGTTAAGGTT$\\
			$CCGGAACCGGAACCGGAACCGGAA$   &   $AAGGCCAAGGCCAAGGCCAAGGCC$\\
			$GGGGGGGGGGGGGGGGGGGGGGGG$   &   $GGTTGGGGTTGGGGTTGGGGTTGG$\\
			$TTGGGGTTGGGGTTGGGGTTGGGG$   &   $GGGGTTGGGGTTGGGGTTGGGGTT$\\
			$CCGGGGCCGGGGCCGGGGCCGGGG$   &   $GGGGCCGGGGCCGGGGCCGGGGCC$\\
			$TTAAGGTTAAGGTTAAGGTTAAGG$   &   $GGAATTGGAATTGGAATTGGAATT$\\
			$AATTGGAATTGGAATTGGAATTGG$   &   $GGTTAAGGTTAAGGTTAAGGTTAA$\\
			$TTTTGGTTTTGGTTTTGGTTTTGG$   &   $GGTTTTGGTTTTGGTTTTGGTTTT$\\
			$CCTTGGCCTTGGCCTTGGCCTTGG$   &   $GGTTCCGGTTCCGGTTCCGGTTCC$\\
			$CCAAGGCCAAGGCCAAGGCCAAGG$   &   $GGAACCGGAACCGGAACCGGAACC$\\
			$AACCGGAACCGGAACCGGAACCGG$   &   $GGCCAAGGCCAAGGCCAAGGCCAA$\\
			$GGCCGGGGCCGGGGCCGGGGCCGG$   &   $AAAAAAAAAAAAAAAAAAAAAAAA$\\
			$TTCCGGTTCCGGTTCCGGTTCCGG$   &   $GGCCTTGGCCTTGGCCTTGGCCTT$\\
			$CCCCGGCCCCGGCCCCGGCCCCGG$   &   $GGCCCCGGCCCCGGCCCCGGCCCC$\\
			$TTAAAATTAAAATTAAAATTAAAA$   &   $AAAATTAAAATTAAAATTAAAATT$\\
			$AATTAAAATTAAAATTAAAATTAA$   &   $TTGGTTTTGGTTTTGGTTTTGGTT$\\
			$TTTTAATTTTAATTTTAATTTTAA$   &   $AATTTTAATTTTAATTTTAATTTT$\\
				\hline
			\end{tabular}
			\label{ta4}
		\end{table}
	
\begin{flushleft}
Table 6 continued.
\end{flushleft}
\begin{tabular}{|c|c|}
\hline
			$CCTTAACCTTAACCTTAACCTTAA$   &   $AATTCCAATTCCAATTCCAATTCC$\\
			$CCGGTTCCGGTTCCGGTTCCGGTT$   &   $TTGGCCTTGGCCTTGGCCTTGGCC$\\
			$TTAATTTTAATTTTAATTTTAATT$   &   $TTTTTTTTTTTTTTTTTTTTTTTT$\\
			$CCTTTTCCTTTTCCTTTTCCTTTT$   &   $TTTTCCTTTTCCTTTTCCTTTTCC$\\
			$CCAATTCCAATTCCAATTCCAATT$   &   $TTAACCTTAACCTTAACCTTAACC$\\
			$AACCTTAACCTTAACCTTAACCTT$   &   $TTCCAATTCCAATTCCAATTCCAA$\\
			$TTCCTTTTCCTTTTCCTTTTCCTT$   &   $AACCAAAACCAAAACCAAAACCAA$\\
			$CCCCTTCCCCTTCCCCTTCCCCTT$   &   $TTCCCCTTCCCCTTCCCCTTCCCC$\\
			$CCAAAACCAAAACCAAAACCAAAA$   &   $AAAACCAAAACCAAAACCAAAACC$\\
			$CCCCAACCCCAACCCCAACCCCAA$   &   $AACCCCAACCCCAACCCCAACCCC$\\
			$CCGGCCCCGGCCCCGGCCCCGGCC$   &   $CCTTCCCCTTCCCCTTCCCCTTCC$\\
			$CCAACCCCAACCCCAACCCCAACC$   &   $CCCCCCCCCCCCCCCCCCCCCCCC$\\
			$CCTTTTGGAAAAGGAAAAGGAAAA$   &   $AAAAGGAAAAGGAAAAGGTTTTCC$\\
			$TTCCTTAAGGAAAAGGAAAAGGAA$   &   $AAGGAAAAGGAAAAGGAATTCCTT$\\
			$CCCCTTGGGGAAGGGGAAGGGGAA$   &   $AAGGGGAAGGGGAAGGGGTTCCCC$\\
			$AACCTTTTGGAATTGGAATTGGAA$   &   $AAGGTTAAGGTTAAGGTTTTCCAA$\\
			$GGCCTTCCGGAACCGGAACCGGAA$   &   $AAGGCCAAGGCCAAGGCCTTCCGG$\\
			$TTTTCCAAAAGGAAAAGGAAAAGG$   &   $GGAAAAGGAAAAGGAAAACCTTTT$\\
			$CCTTCCGGAAGGGGAAGGGGAAGG$   &   $GGAAGGGGAAGGGGAAGGCCTTCC$\\
			$TTCCCCAAGGGGAAGGGGAAGGGG$   &   $GGGGAAGGGGAAGGGGAACCCCTT$\\
			$CCCCCCGGGGGGGGGGGGGGGGGG$   &   $GGGGGGGGGGGGGGGGGGCCCCCC$\\
			$AACCCCTTGGGGTTGGGGTTGGGG$   &   $GGGGTTGGGGTTGGGGTTCCCCAA$\\
			$GGCCCCCCGGGGCCGGGGCCGGGG$   &   $GGGGCCGGGGCCGGGGCCCCCCGG$\\
			$AATTCCTTAAGGTTAAGGTTAAGG$   &   $GGAATTGGAATTGGAATTCCTTAA$\\
			$TTAACCAATTGGAATTGGAATTGG$   &   $GGTTAAGGTTAAGGTTAACCAATT$\\
			$CCAACCGGTTGGGGTTGGGGTTGG$   &   $GGTTGGGGTTGGGGTTGGCCAACC$\\
			$AAAACCTTTTGGTTTTGGTTTTGG$   &   $GGTTTTGGTTTTGGTTTTCCAAAA$\\
			$GGAACCCCTTGGCCTTGGCCTTGG$   &   $GGTTCCGGTTCCGGTTCCCCAAGG$\\
			$GGTTCCCCAAGGCCAAGGCCAAGG$   &   $GGAACCGGAACCGGAACCCCTTGG$\\
			$TTGGCCAACCGGAACCGGAACCGG$   &   $GGCCAAGGCCAAGGCCAACCGGTT$\\
			$CCGGCCGGCCGGGGCCGGGGCCGG$   &   $GGCCGGGGCCGGGGCCGGCCGGCC$\\
			$AAGGCCTTCCGGTTCCGGTTCCGG$   &   $GGCCTTGGCCTTGGCCTTCCGGAA$\\
			$GGGGCCCCCCGGCCCCGGCCCCGG$   &   $GGCCCCGGCCCCGGCCCCCCGGGG$\\
			$TTTTTTAAAAAAAAAAAAAAAAAA$   &   $AAAAAAAAAAAAAAAAAATTTTTT$\\
			$AATTTTTTAAAATTAAAATTAAAA$   &   $AAAATTAAAATTAAAATTTTTTAA$\\
			$CCAATTGGTTAAGGTTAAGGTTAA$   &   $AATTGGAATTGGAATTGGTTAACC$\\
			$TTAATTAATTAAAATTAAAATTAA$   &   $AATTAAAATTAAAATTAATTAATT$\\
			$AAAATTTTTTAATTTTAATTTTAA$   &   $AATTTTAATTTTAATTTTTTAAAA$\\
			$GGAATTCCTTAACCTTAACCTTAA$   &   $AATTCCAATTCCAATTCCTTAAGG$\\
			$CCTTAAGGAATTGGAATTGGAATT$   &   $TTAAGGTTAAGGTTAAGGAATTCC$\\
			$CCCCAAGGGGTTGGGGTTGGGGTT$   &   $TTGGGGTTGGGGTTGGGGAACCCC$\\
			$TTCCAAAAGGTTAAGGTTAAGGTT$   &   $TTGGAATTGGAATTGGAAAACCTT$\\
			$AACCAATTGGTTTTGGTTTTGGTT$   &   $TTGGTTTTGGTTTTGGTTAACCAA$\\
			$GGCCAACCGGTTCCGGTTCCGGTT$   &   $TTGGCCTTGGCCTTGGCCAACCGG$\\
			$TTTTAAAAAATTAAAATTAAAATT$   &   $TTAAAATTAAAATTAAAAAATTTT$\\
			$AATTAATTAATTTTAATTTTAATT$   &   $TTAATTTTAATTTTAATTAATTAA$\\
			$CCAAAAGGTTTTGGTTTTGGTTTT$   &   $TTTTGGTTTTGGTTTTGGAAAACC$\\
			$TTAAAAAATTTTAATTTTAATTTT$   &   $TTTTAATTTTAATTTTAAAAAATT$\\
			\hline
			\end{tabular}
			
			\begin{flushleft}
			Table 6 continued.
			\end{flushleft}
			\begin{tabular}{|c|c|}
			\hline
			$GGAAAACCTTTTCCTTTTCCTTTT$   &   $TTTTCCTTTTCCTTTTCCAAAAGG$\\
			$GGTTAACCAATTCCAATTCCAATT$   &   $TTAACCTTAACCTTAACCAATTGG$\\
			$CCGGAAGGCCTTGGCCTTGGCCTT$   &   $TTCCGGTTCCGGTTCCGGAAGGCC$\\
			$TTGGAAAACCTTAACCTTAACCTT$   &   $TTCCAATTCCAATTCCAAAAGGTT$\\
			$AAGGAATTCCTTTTCCTTTTCCTT$   &   $TTCCTTTTCCTTTTCCTTAAGGAA$\\
			$GGGGAACCCCTTCCCCTTCCCCTT$   &   $TTCCCCTTCCCCTTCCCCAAGGGG$\\
			$GGTTTTCCAAAACCAAAACCAAAA$   &   $AAAACCAAAACCAAAACCTTTTGG$\\
			$CCGGTTGGCCAAGGCCAAGGCCAA$   &   $AACCGGAACCGGAACCGGTTGGCC$\\
			$AAGGTTTTCCAATTCCAATTCCAA$   &   $AACCTTAACCTTAACCTTTTGGAA$\\
			$TTGGTTAACCAAAACCAAAACCAA$   &   $AACCAAAACCAAAACCAATTGGTT$\\
			$GGGGTTCCCCAACCCCAACCCCAA$   &   $AACCCCAACCCCAACCCCTTGGGG$\\
			$CCTTGGGGAACCGGAACCGGAACC$   &   $CCAAGGCCAAGGCCAAGGGGTTCC$\\
			$CCCCGGGGGGCCGGGGCCGGGGCC$   &   $CCGGGGCCGGGGCCGGGGGGCCCC$\\
			$AACCGGTTGGCCTTGGCCTTGGCC$   &   $CCGGTTCCGGTTCCGGTTGGCCAA$\\
			$TTCCGGAAGGCCAAGGCCAAGGCC$   &   $CCGGAACCGGAACCGGAAGGCCTT$\\
			$GGCCGGCCGGCCCCGGCCCCGGCC$   &   $CCGGCCCCGGCCCCGGCCGGCCGG$\\
			$AATTGGTTAACCTTAACCTTAACC$   &   $CCAATTCCAATTCCAATTGGTTAA$\\
			$CCAAGGGGTTCCGGTTCCGGTTCC$   &   $CCTTGGCCTTGGCCTTGGGGAACC$\\
			$AAAAGGTTTTCCTTTTCCTTTTCC$   &   $CCTTTTCCTTTTCCTTTTGGAAAA$\\
			$TTAAGGAATTCCAATTCCAATTCC$   &   $CCTTAACCTTAACCTTAAGGAATT$\\
			$GGAAGGCCTTCCCCTTCCCCTTCC$   &   $CCTTCCCCTTCCCCTTCCGGAAGG$\\
			$TTTTGGAAAACCAAAACCAAAACC$   &   $CCAAAACCAAAACCAAAAGGTTTT$\\
			$GGTTGGCCAACCCCAACCCCAACC$   &   $CCAACCCCAACCCCAACCGGTTGG$\\
			$CCGGGGGGCCCCGGCCCCGGCCCC$   &   $CCCCGGCCCCGGCCCCGGGGGGCC$\\
			$AAGGGGTTCCCCTTCCCCTTCCCC$   &   $CCCCTTCCCCTTCCCCTTGGGGAA$\\
			$TTGGGGAACCCCAACCCCAACCCC$   &   $CCCCAACCCCAACCCCAAGGGGTT$\\
			$GGGGGGCCCCCCCCCCCCCCCCCC$   &   $CCCCCCCCCCCCCCCCCCGGGGGG$\\
			\hline
		\end{tabular}

	\section{Conclusion}\label{s9}
	In this paper, we have studied the structure of double cyclic codes of length $(\gamma,\delta)$ over the ring $\mathbb F_q+u\mathbb F_q, u^2=0$ when both $\gamma$ and $\delta$ are coprime with $q$. We also discussed the duality of these codes. We have studied the reversibility of double cyclic codes of length $(\gamma,\delta=(kp+1)\gamma)$, where both $\gamma$ and $\delta$ are coprime with $q$. Moreover we have also studied the reversible-complement codes over the ring $\mathbb F_4+u\mathbb F_4, u^2=0$ that are suitable for DNA code construction. Further we have provided some examples of these codes. For future studies, it will be interesting to study double cyclic codes over different rings with different lengths and extending these studies to DNA codes construction and DNA computing.

\end{document}